\documentclass{birkjour}
\usepackage{amsmath}
\usepackage{amsfonts}
\usepackage{amssymb}
\usepackage{graphicx}
\usepackage[all]{xy}%

 \newtheorem{thm}{Theorem}[section]
 \newtheorem{cor}[thm]{Corollary}
 
 \newtheorem{prop}[thm]{Proposition}
 \theoremstyle{definition}
 \newtheorem{defn}[thm]{Definition}
 \theoremstyle{remark}
 \newtheorem{rem}[thm]{Remark}
 \newtheorem{ex}[thm]{Example}
 \numberwithin{equation}{section}

\DeclareMathOperator{\sech}{sech}

\begin{document}
%
%
%
%
%
%
%
%
%

\title[Transmutations and SPPS in eigenvalue problems]{Transmutations and spectral parameter\\ power series in eigenvalue problems }

\author{Vladislav V. Kravchenko}
\address{%
Department of Mathematics, CINVESTAV del IPN, Unidad Queretaro,\\
Libramiento Norponiente No. 2000, Fracc. Real de Juriquilla,\\
Queretaro, Qro. C.P. 76230 MEXICO}
\email{vkravchenko@math.cinvestav.edu.mx}
\thanks{Research was supported by CONACYT, Mexico.
Research of second named author was supported by DFFD, Ukraine (GP/F32/030)
and by SNSF, Switzerland (JRP IZ73Z0 of SCOPES 2009--2012).}

\author{Sergii M. Torba}
\address{%
Department of Mathematics, CINVESTAV del IPN, Unidad Queretaro,\\
Libramiento Norponiente No. 2000, Fracc. Real de Juriquilla,\\
Queretaro, Qro. C.P. 76230 MEXICO}
\email{storba@math.cinvestav.edu.mx}

\subjclass{Primary 34B24, 34L16, 65L15, 81Q05, 81Q60; Secondary 34L25, 34L40}

\keywords{Sturm-Liouville operator, Sturm-Liouville problem, complex eigenvalue, transmutation operator, transformation operator, Schr\"{o}dinger operator, spectral parameter power series, Darboux transformation, quantum well, scalar potential}

\date{February 29, 2012}
\dedicatory{Dedicated to 70th birthday anniversary of Prof. Dr. Vladimir S. Rabinovich.}

\begin{abstract}
We give an overview of recent developments in Sturm-Liouville theory
concerning operators of transmutation (transformation) and spectral parameter
power series (SPPS). The possibility to write down the dispersion
(characteristic) equations corresponding to a variety of spectral problems
related to Sturm-Liouville equations in an analytic form is an attractive
feature of the SPPS method. It is based on a computation of certain systems of
recursive integrals. Considered as families of functions these systems are
complete in the $L_{2}$-space and result to be the images of the nonnegative
integer powers of the independent variable under the action of a corresponding
transmutation operator. This recently revealed property of the Delsarte
transmutations opens the way to apply the transmutation operator even when its
integral kernel is unknown and gives the possibility to obtain further
interesting properties concerning the Darboux transformed Schr\"{o}dinger operators.

We introduce the systems of recursive integrals and the SPPS approach, explain
some of its applications to spectral problems with numerical illustrations,
give the definition and basic properties of transmutation operators, introduce
a parametrized family of transmutation operators, study their mapping
properties and construct the transmutation operators for Darboux transformed
Schr\"{o}dinger operators.

\end{abstract}

\maketitle

\section{Introduction}

Transmutation operators also called operators of transformation are a widely
used tool in the theory of linear differential equations (see, e.g.,
\cite{Gilbert}, \cite{Carroll}, \cite{LevitanInverse}, \cite{Marchenko},
\cite{Trimeche} and the recent review \cite{Sitnik}). It is well known that
under certain quite general conditions the transmutation operator transmuting
the operator $A=-\frac{d^{2}}{dx^{2}}+q(x)$ into $B=-\frac{d^{2}}{dx^{2}}$ is
a Volterra integral operator with good properties. Its integral kernel can be
obtained as a solution of a certain Goursat problem for the Klein-Gordon
equation with a variable coefficient. In particular, the elementary solutions
of the equation $Bv=\lambda v$ are transformed into the solutions of the
equation $Au=\lambda u$. If the integral kernel of the transmutation operator
is unknown, and usually this is the case, there is no way to apply it to an
arbitrary smooth function. This obstacle strongly restricts the application of
the transmutation operators confining it to purely theoretical purposes.

Recently, in \cite{CKT} a relation of the transmutation operators with another
fundamental object of the Sturm-Liouville theory was revealed. Sometimes this
object is called the $L$-basis \cite{Fage} where $L$ refers to a corresponding
linear ordinary differential operator. The $L$-basis is an infinite family of
functions $\left\{  \varphi_{k}\right\}  _{k=0}^{\infty}$ such that
$L\varphi_{k}=0$ for $k=0,1$, $L\varphi_{k}=k(k-1)\varphi_{k-2}$, for
$k=2,3,\ldots$ and all $\varphi_{k}$ satisfy certain prescribed initial
conditions. In \cite{KrCV08}, \cite{APFT}, \cite{KrPorter2010} it was shown
that the $L$-basis naturally arises in a representation of the solutions of
the Sturm-Liouville equation in terms of powers of the spectral parameter. The
approach based on such representation is called the spectral parameter power
series (SPPS) method. The functions $\varphi_{k}$ which constitute the
$L$-basis appear as the expansion coefficients in the SPPS. In \cite{KrCV08},
\cite{APFT} and \cite{KrPorter2010} convenient representations for their
practical computation were proposed which converted the SPPS method into an
efficient and highly competitive technique for solving a variety of spectral
and scattering problems related to Sturm-Liouville equations (see
\cite{CKKO2009}, \cite{CKOR}, \cite{KKRosu}, \cite{KiraRosu2010},
\cite{KrPorter2010}, \cite{KV2011}). The above mentioned relation between the
transmutation operators and the functions $\varphi_{k}$ called in the present
paper the recursive integrals consists in the fact established in \cite{CKT}
that for every system $\left\{  \varphi_{k}\right\}  _{k=0}^{\infty}$ there
exists a transmutation operator $\mathbf{T}$ such that $\mathbf{T}%
[x^{k}]=\varphi_{k}$, i.e., the functions $\varphi_{k}$ are the images of the
usual powers of the independent variable. Moreover, it was shown how this
operator can be constructed and how it is related to the \textquotedblleft
canonical\textquotedblright\ transformation operator considered, e.g., in
\cite[Chapter 1]{Marchenko}. This result together with the practical formulas
for calculating the functions $\varphi_{k}$ makes it possible to apply the
transmutation technique even when the integral kernel of the operator is
unknown. Indeed, now it is easy to apply the transmutation operator to any
function approximated by a polynomial.

Deeper understanding of the mapping properties of the transmutation operators
led us in \cite{KrT2012} to the explicit construction of the transmutation
operator for a Darboux transformed Schr\"{o}dinger operator by a known
transmutation operator for the original Schr\"{o}dinger operator as well as to
several interesting relations between the two transmutation operators. These
relations also allowed us to prove the main theorem on the transmutation
operators under a weaker condition than it was known before (not requiring the
continuous differentiablity of the potential in the Schr\"{o}dinger operator).

In the present paper we overview the recent results related to the SPPS
approach explaining and illustrating its main advantage, the possibility to
write down in an analytic form the characteristic equation of the spectral
problem. This equation can be approximated in different ways, and its
solutions give us the eigenvalues of the problem. In other words the
eigenvalue problem reduces to computation of zeros of a certain complex
analytic function given by its Taylor series whose coefficients are obtained
as simple linear combinations of the values of the functions $\varphi_{k}$ at
a given point. We discuss different applications of the SPPS method and give
the results of some comparative numerical calculations.

Following \cite{CKT} and \cite{KrT2012} we introduce a parametrized family of
transmutation operators and study their mapping properties,we give an explicit
representation for the kernel of the transmutation operator corresponding to
the Darboux transformed potential in terms of the transmutation kernel for its
superpartner (Theorem \ref{Th T2Volterra}). Moreover, this result leads to
interesting commutation relations between the two transmutation operators
(Corollary \ref{Cor Commutation Relations}) which in their turn allow us to
obtain a transmutation operator for the one-dimensional Dirac system with a
scalar potential as well as to prove the main property of the transmutation
operator under less restrictive conditions than it has been proved until now.
We give several examples of explicitly constructed kernels of transmutation
operators. It is worth mentioning that in the literature there are very few
explicit examples and even in the case when $q$ is a constant such kernel was
presented recently in \cite{CKT}. The results discussed in the present paper
allow us to enlarge considerably the list of available examples and give a
relatively simple tool for constructing Darboux related sequences of the
transmutation kernels.

\section{Recursive integrals: a question on the completeness}

Let $f\in C^{2}(a,b)\cap C^{1}[a,b]$ be a complex valued function and
$f(x)\neq0$ for any $x\in[a,b]$. The interval $(a,b)$ is assumed being
finite. Let us consider the following functions%
\begin{multline}
X^{(0)}(x)\equiv1,\qquad X^{(n)}(x)=n\int_{x_{0}}^{x}X^{(n-1)}(s)\left(
f^{2}(s)\right)  ^{(-1)^{n}}\,\mathrm{d}s,\\ x_{0}\in[a,b],\quad
n=1,2,\ldots. \label{Xn}%
\end{multline}

We pose the following questions. \textbf{Is the family of functions $\left\{
X^{(n)}\right\}  _{n=0}^{\infty}$ complete let us say in 
$L_{2}(a,b)$? What about the completeness of $\left\{  X^{(2n)}%
\right\}  _{n=0}^{\infty}$ or $\left\{  X^{(2n+1)}\right\}
_{n=0}^{\infty}$?}

The following example shows that both questions are meaningful and natural.
\begin{ex}
Let $f\equiv1$, $a=0$, $b=1$. Then it is easy to see that choosing $x_{0}=0$
we have
$X^{(0)}(x)=1,\ X^{(1)}(x)=x,\ X^{(2)}(x)=x^{2},\ X^{(3)}(x)=x^{3},\ldots$.
Thus, the family of functions $\left\{  X^{(n)}\right\}  _{n=0}^{\infty}$ is
complete in $L_{2}(0,1)$. Moreover, both $\left\{  X^{(2n)}\right\}
_{n=0}^{\infty}$ and $\left\{  X^{(2n+1)}\right\}  _{n=0}^{\infty}$ are
complete in $L_{2}(0,1)$ as well.

If instead of $a=0$ we choose $a=-1$ then \textbf{$\left\{  X^{(n)}\right\}
_{n=0}^{\infty}$ is still complete in $L_{2}(-1,1)$ but neither $\left\{
X^{(2n)}\right\}  _{n=0}^{\infty}$ nor $\left\{  X^{(2n+1)}\right\}
_{n=0}^{\infty}$.}
\end{ex}

Together with the family of functions $\left\{  X^{(n)}\right\}
_{n=0}^{\infty}$ we consider also another similarly defined family of
functions $\big\{  \widetilde{X}^{(n)}\big\}  _{n=0}^{\infty}$,
\begin{multline}
\widetilde{X}^{(0)}\equiv1,\qquad\widetilde{X}^{(n)}(x)=n\int_{x_{0}}%
^{x}\widetilde{X}^{(n-1)}(s)\left(  f^{2}(s)\right)  ^{(-1)^{n-1}}%
\,\mathrm{d}s,\\ x_{0}\in[a,b],\quad n=1,2,\ldots. \label{Xtiln}%
\end{multline}

\begin{rem}
As we show below the introduced families of functions are closely related to
the one-dimensional Schr\"{o}dinger equations of the form $u^{\prime\prime
}-qu=\lambda u$ where $q$ is a complex-valued continuous function. Slightly
more general families of functions can be studied in relation to
Sturm-Liouville equations of the form $(py^{\prime})^{\prime}+qy=\lambda ry$.
Their definition based on a corresponding recursive integration procedure is
given in \cite{APFT}, \cite{KrPorter2010}, \cite{KKRosu}.
\end{rem}

We introduce the infinite system of functions $\left\{  \varphi_{k}\right\}
_{k=0}^{\infty}$ defined as follows
\begin{equation}
\varphi_{k}(x)=%
\begin{cases}
f(x)X^{(k)}(x), & k\text{\ odd,}\\
f(x)\widetilde{X}^{(k)}(x), & k\text{\ even.}%
\end{cases}
\label{phik}%
\end{equation}

The system (\ref{phik}) is closely related to the notion of the $L$-basis
introduced and studied in \cite{Fage}. Here the letter $L$ corresponds to a
linear ordinary differential operator.

Together with the system of functions (\ref{phik}) we define the functions
$\{\psi_{k}\}_{k=0}^{\infty}$ using the \textquotedblleft second
half\textquotedblright\ of the recursive integrals \eqref{Xn} and
\eqref{Xtiln},
\begin{equation}
\psi_{k}(x)=%
\begin{cases}
\dfrac{\widetilde{X}^{(k)}(x)}{f(x)}, & k\text{\ odd,}\\
\dfrac{X^{(k)}(x)}{f(x)}, & k\text{\ even.}
\end{cases}
\label{psik}%
\end{equation}
The following result obtained in \cite{KrCV08} (for additional details and
simpler proof see \cite{APFT} and \cite{KrPorter2010}) establishes the
relation of the system of functions $\left\{  \varphi_{k}\right\}
_{k=0}^{\infty}$ and $\left\{  \psi_{k}\right\}
_{k=0}^{\infty}$ to the Sturm-Liouville equation.

\begin{thm}
[\cite{KrCV08}]\label{ThGenSolSturmLiouville} Let $q$ be a continuous complex
valued function of an independent real variable $x\in[a,b]$ and
$\lambda$ be an arbitrary complex number. Suppose there exists a solution $f$
of the equation
\begin{equation}
f^{\prime\prime}-qf=0\label{SLhom}%
\end{equation}
on $(a,b)$ such that $f\in C^{2}(a,b)\cap C^{1}[a,b]$ and $f(x)\neq0$\ for any
$x\in[a,b]$. Then the general solution $u\in C^{2}(a,b)\cap C^{1}[a,b]$
of the equation
\begin{equation}
u^{\prime\prime}-qu=\lambda u\label{SLlambda}%
\end{equation}
on $(a,b)$ has the form%
\[
u=c_{1}u_{1}+c_{2}u_{2}%
\]
where $c_{1}$ and $c_{2}$ are arbitrary complex constants,
\begin{equation}
u_{1}=\sum_{k=0}^{\infty}\frac{\lambda^{k}}{(2k)!}\varphi_{2k}\qquad
\text{and}\qquad u_{2}=\sum_{k=0}^{\infty}\frac{\lambda^{k}}{(2k+1)!}%
\varphi_{2k+1}\label{u1u2}%
\end{equation}
and both series converge uniformly on $[a,b]$ together with the series of the
first derivatives which have the form%
\begin{multline}\label{du1du2}
u_{1}^{\prime}=f^{\prime}+\sum_{k=1}^{\infty}\frac{\lambda^{k}}{(2k)!}\left(
\frac{f^{\prime}}{f}\varphi_{2k}+2k\,\psi_{2k-1}\right)  \qquad\text{and}\\
u_{2}^{\prime}=\sum_{k=0}^{\infty}\frac{\lambda^{k}}{(2k+1)!}\left(
\frac{f^{\prime}}{f}\varphi_{2k+1}+\left(  2k+1\right)  \psi_{2k}\right).
\end{multline}
The series of the second derivatives converge uniformly on any segment\linebreak
$[a_{1},b_{1}]\subset(a,b)$.
\end{thm}

The representation (\ref{u1u2}) offers the linearly independent solutions of
(\ref{SLlambda}) in the form of spectral parameter power series (SPPS). The
possibility to represent solutions of the Sturm-Liouville equation in such
form is by no means a novelty, though it is not a widely used tool (in fact,
besides the work reviewed below and in \cite{KKRosu} we are able to mention
only \cite[Sect. 10]{Bellman}, \cite{DelsarteLions1956} and the recent paper
\cite{KostenkoTeschl}) and to our best knowledge for the first time it was
applied for solving spectral problems in \cite{KrPorter2010}. The reason of
this underuse of the SPPS lies in the form in which the expansion coefficients
were sought. Indeed, in previous works the calculation of coefficients was
proposed in terms of successive integrals with the kernels in the form of
iterated Green functions (see \cite[Sect. 10]{Bellman}). This makes any
computation based on such representation difficult, less practical and even
proofs of the most basic results like, e.g., the uniform convergence of the
spectral parameter power series for any value of $\lambda\in\mathbb{C}$
(established in Theorem \ref{ThGenSolSturmLiouville}) are not an easy task.
For example, in \cite[p. 16]{Bellman} the parameter $\lambda$ is assumed to be
small and no proof of convergence is given.

The way of how the expansion coefficients in (\ref{u1u2}) are calculated
according to (\ref{Xn}), (\ref{Xtiln}) is relatively simple and
straightforward, this is why the estimation of the rate of convergence of the
series (\ref{u1u2}) presents no difficulty, see \cite{KrPorter2010}. Moreover,
in \cite{CamposKr} a discrete analogue of Theorem \ref{ThGenSolSturmLiouville}
was established and the discrete analogues of the series (\ref{u1u2}) resulted
to be finite sums.

Another crucial feature of the introduced representation of the expansion
coefficients in (\ref{u1u2}) consists in the fact that not only these
coefficients (denoted by $\varphi_{k}$ in (\ref{phik})) are required for
solving different spectral problems related to the Sturm-Liouville equation.
Indeed, the functions $\widetilde{X}^{(2k+1)}$ and $X^{(2k)}$, $k=0,1,2,\ldots
$ do not participate explicitly in the representation (\ref{u1u2}). Nevertheless,
together with the functions $\varphi_{k}$ they appear in the representation \eqref{du1du2} of the derivatives of the solutions and therefore also in
characteristic equations corresponding to the spectral problems.

In the present work we also overview another approach developed in
\cite{KrCMA2011}, \cite{CKM}, \cite{CKT} and \cite{KrT2012} where the formal
powers (\ref{Xn}) and (\ref{Xtiln}) were considered as infinite families of
functions intimately related to the corresponding Sturm-Liouville operator. As
we show this leads to a deeper understanding of the transmutation operators
\cite{Gilbert}, \cite{Carroll} also known as transformation operators
\cite{LevitanInverse}, \cite{Marchenko}. Indeed, the functions $\varphi
_{k}(x)$ result to be the images of the powers $x^{k}$ under the action of a
corresponding transmutation operator \cite{CKT}. This makes it possible to
apply the transmutation operator even when the operator itself is unknown (and
this is the usual situation -- very few explicit examples are available) due
to the fact that its action on every polynomial is known. This result was used
in \cite{CKM} and \cite{CKT} \ to prove the completeness (Runge-type
approximation theorems) for families of solutions of two-dimensional
Schr\"{o}dinger and Dirac equations with variable complex-valued coefficients.

\begin{rem}
\label{RemInitialValues}It is easy to see that by definition the solutions
$u_{1}$ and $u_{2}$ from (\ref{u1u2}) satisfy the following initial
conditions
\begin{align*}
u_{1}(x_{0})  &  =f(x_{0}), & u_{1}^{\prime}(x_{0})  &  =f^{\prime}(x_{0}),\\
u_{2}(x_{0})  &  =0, & u_{2}^{\prime}(x_{0})  &  =1/f(x_{0}).
\end{align*}

\end{rem}

\begin{rem}
\label{RemarkNonVanish} It is worth mentioning that in the regular case the
existence and construction of the required $f$ presents no difficulty. Let $q$
be real valued and continuous on $[a,b]$. Then (\ref{SLhom}) possesses two
linearly independent regular solutions\/ $v_{1}$ and $v_{2}$ whose zeros
alternate. Thus one may choose $f=v_{1}+iv_{2}$. Moreover, for the
construction of $v_{1}$ and $v_{2}$ in fact the same SPPS method may be used
\cite{KrPorter2010}.
\end{rem}

Theorem \ref{ThGenSolSturmLiouville} together with the results on the
completeness of Sturm-Liouville eigenfunctions and generalized eigenfunctions
\cite{Marchenko} implies the validity of the following two statements. For
their detailed proofs we refer to \cite{KrCMA2011} and \cite{KMoT} respectively.

\begin{thm}
[\cite{KrCMA2011}]Let $(a,b)$ be a finite interval and $f\in C^{2}(a,b)\cap
C^{1}[a,b]$ be a complex valued function such that $f(x)\neq0$ for any
$x\in[a,b]$.

If $x_{0}=a$ (or $x_{0}=b$) then each of the four systems of functions
$\big\{  X^{(2n)}\big\}_{n=0}^{\infty}$, $\big\{
X^{(2n+1)}\big\}_{n=0}^{\infty}$, $\big\{  \widetilde{X}^{(2n)}\big\}
_{n=0}^{\infty}$, $\big\{  \widetilde{X}^{(2n+1)}\big\}  _{n=0}^{\infty}$
is complete in $L_{2}(a,b)$.

If $x_{0}$ is an arbitrary point of the interval $(a,b)$ then each of the
following two combined systems of functions $\big\{  \widetilde{X}%
^{(2n)}\big\}  _{n=0}^{\infty}\cup\big\{  X^{(2n+1)}\big\}
_{n=0}^{\infty}$ and $\big\{  \widetilde{X}^{(2n+1)}\big\}  _{n=0}^{\infty
}\cup\big\{  X^{(2n)}\big\}  _{n=0}^{\infty}$ is complete in $L_{2}(a,b)$.
\end{thm}

\begin{thm}
[\cite{KMoT}]\label{ThComplMaxNorm}Let $f$ satisfy the conditions of the
preceding theorem and $\left\{  \varphi_{k}\right\}  _{k=0}^{\infty}$ be the
system of functions defined by \eqref{phik} with $x_{0}$ being an arbitrary
point of the interval $[a,b]$. Then for any complex valued continuous,
piecewise continuously differentiable function $h$ defined on $[a,b]$ and for
any $\varepsilon>0$ there exists such $N\in\mathbb{N}$ and such complex
numbers $\alpha_{k}$, $k=0,1,\ldots,N$ that 
\[\max_{x\in[a,b]}\left\vert
h(x)-\sum_{k=0}^{N}\alpha_{k}\varphi_{k}\right\vert <\varepsilon.\]
\end{thm}

\section{Dispersion relations for spectral problems and approximate solutions}

The SPPS representation (\ref{u1u2}) for solutions of the Sturm-Liouville
equation (\ref{SLlambda}) is very convenient for writing down the dispersion
(characteristic) relations in an analytical form. This fact was used in
\cite{CKOR}, \cite{KKRosu}, \cite{KiraRosu2010}, \cite{KrPorter2010},
\cite{KV2011} for approximating solutions of different eigenvalue problems.
Here in order to explain this we consider two examples: the Sturm-Liouville
problem and the quantum-mechanical eigenvalue problem. As the performance of
the SPPS method in application to classical (regular and singular)
Sturm-Liouville problems was studied in detail in \cite{KrPorter2010} here we
consider the Sturm-Liouville problems with boundary conditions which depend on
the spectral parameter $\lambda$. This situation occurs in many physical
models (see, e.g., \cite{BenAmara,Chanane2008,CodeBrowne2005,CoskunBayram2005,
Fulton77,Walter} and references therein) and is considerably more difficult
from the computational point of view. Moreover, as we show in this section the
SPPS method is applicable to models admitting complex eigenvalues - an
important advantage in comparison with the best purely numerical techniques
all of them being based on the shooting method.

Consider the equation $u^{\prime\prime}-qu=\lambda u$ together with the
boundary conditions
\begin{equation}
u(a)\cos\alpha+u^{\prime}(a)\sin\alpha=0, \label{at0}%
\end{equation}%
\begin{equation}
\beta_{1}u(b)-\beta_{2}u^{\prime}(b)=\phi(\lambda)\big(  \beta_{1}^{\prime
}u(b)-\beta_{2}^{\prime}u^{\prime}(b)\big)  , \label{SL53}%
\end{equation}
where $\alpha$ is an arbitrary complex number, $\phi$ is a complex-valued
function of the variable $\lambda$ and $\beta_{1}$, $\beta_{2}$, $\beta
_{1}^{\prime}$, $\beta_{2}^{\prime}$ are complex numbers. For some special
forms of the function $\phi$ such as $\phi(\lambda)=\lambda$ or $\phi
(\lambda)=\lambda^{2}+c_{1}\lambda+c_{2}$, results were obtained
\cite{CodeBrowne2005}, \cite{Walter} concerning the regularity of the problem;
we will not dwell upon the details.\ Notice that the SPPS approach is
applicable as well to a more general Sturm-Liouville equation $(pu^{\prime
})^{\prime}+qu=\lambda ru$. For the corresponding details we refer to
\cite{KKRosu} and \cite{KrPorter2010}.

For simplicity, let us suppose that $\alpha=0$ and hence the condition
\eqref{at0} becomes $u(a)=0$. Then choosing the initial integration point in
\eqref{Xn} and \eqref{Xtiln} as $x_{0}=a$ and taking into account Remark
\ref{RemInitialValues} we obtain that if an eigenfunction exists it
necessarily coincides with $u_{2}$ up to a multiplicative constant. In this
case condition \eqref{SL53} becomes equivalent to the equality
\cite{KrPorter2010}, \cite{KKRosu}\looseness=-1

\begin{equation}
\bigl(f(b)\phi_{1}(\lambda)-f'(b)\phi_{2}(\lambda)\bigr)
\sum_{k=0}^{\infty}
\frac{\lambda^{k}}{(2k+1)!}X^{(2k+1)}(b)-
\frac{\phi_{2}(\lambda)}{f(b)}
\sum_{k=0}^{\infty}
\frac{\lambda^{k}}{(2k)!}X^{(2k)}(b)=0, \label{eqSLparam}%
\end{equation}
where $\phi_{1,2}(\lambda)=\beta_{1,2}-\beta_{1,2}^{\prime}\phi(\lambda)$.
This is the characteristic equation of the considered spectral problem.
Calculation of eigenvalues given by (\ref{eqSLparam}) is especially simple in
the case of $\phi$ being a polynomial of $\lambda$. Precisely this particular
situation was considered in all of the above-mentioned references concerning
Sturm-Liouville problems with spectral parameter dependent boundary
conditions. In any case the knowledge of an explicit characteristic equation
(\ref{eqSLparam}) for the spectral problem makes possible its accurate and
efficient solution. For this the infinite sums in (\ref{eqSLparam}) are
truncated after a certain $N\in\mathbb{N}$. The paper \cite{KrPorter2010}
contains several numerical tests corresponding to a variety of computationally
difficult problems. All they reveal an excellent performance of the SPPS
method. We do not review them here referring the interested reader to
\cite{KrPorter2010}. Instead we consider another interesting example from
\cite{KKRosu}, a Sturm-Liouville problem admitting complex eigenvalues.

\begin{ex}
\label{Example1}Consider the equation (\ref{SLlambda}) with $q\equiv0$ on the
interval $(0,\pi)$ with the boundary conditions $u(0)=0$ and $u(\pi
)=-\lambda^{2}u(\pi)$. The exact eigenvalues of the problem are $\lambda
_{n}=n^{2}$ together with the purely imaginary numbers $\lambda_{\pm}=\pm i$.
Application of the SPPS method with $N=100$ and $3000$ interpolating points
(used for representing the integrands as splines) delivered the following
results $\lambda_{1}=1$, $\lambda_{2}=4.0000000000007$, $\lambda
_{3}=9.00000000001$, $\lambda_{4}=15.99999999996$, $\lambda_{5}=25.000000002$,
$\lambda_{6}=35.99999997$, $\lambda_{7}=49.0000004$, $\lambda_{8}=63.9999994$,
$\lambda_{9}=80.9996$, $\lambda_{10}=100.02$ and $\lambda_{\pm}=\pm i$. Thus,
the complex eigenvalues are as easily and accurately detected by the SPPS
method as the real eigenvalues. Note that for a better accuracy in calculation
of higher eigenvalues of a Sturm-Liouville problem an additional simple
shifting procedure described in \cite{KrPorter2010} and based on the
representation of solutions not as series in powers of $\lambda$ but in powers
of $(\lambda-\lambda_{0})$ is helpful. We did not apply it here and hence the
accuracy of the calculated value of $\lambda_{10}$ is considerably worse than
the accuracy of the first calculated eigenvalues which in general can be
improved by means of the mentioned shifting procedure.

Figures 1-3 give us an idea about the stability of the computed eigenvalues
when $N$ increases. In Fig. 1 we plot $\lambda_{1}$ and $\lambda_{2}$ computed
with $N=14,16,\ldots,120$. Figure 2 shows $\lambda_{3}$ computed with
$N=24,30,\ldots,140$ and Figure 3 shows $\lambda_{4}$ computed with
$N=40,50,\ldots,140$ \ Similar plots can be done for calculated higher
eigenvalues. In all cases the computed eigenvalues reveal a remarkable
stability when $N$ increases.
\begin{figure}
[ptb]
\begin{center}
\includegraphics[
natheight=4.379400in,
natwidth=5.862600in,
height=3.0044in,
width=4.0136in
]%
{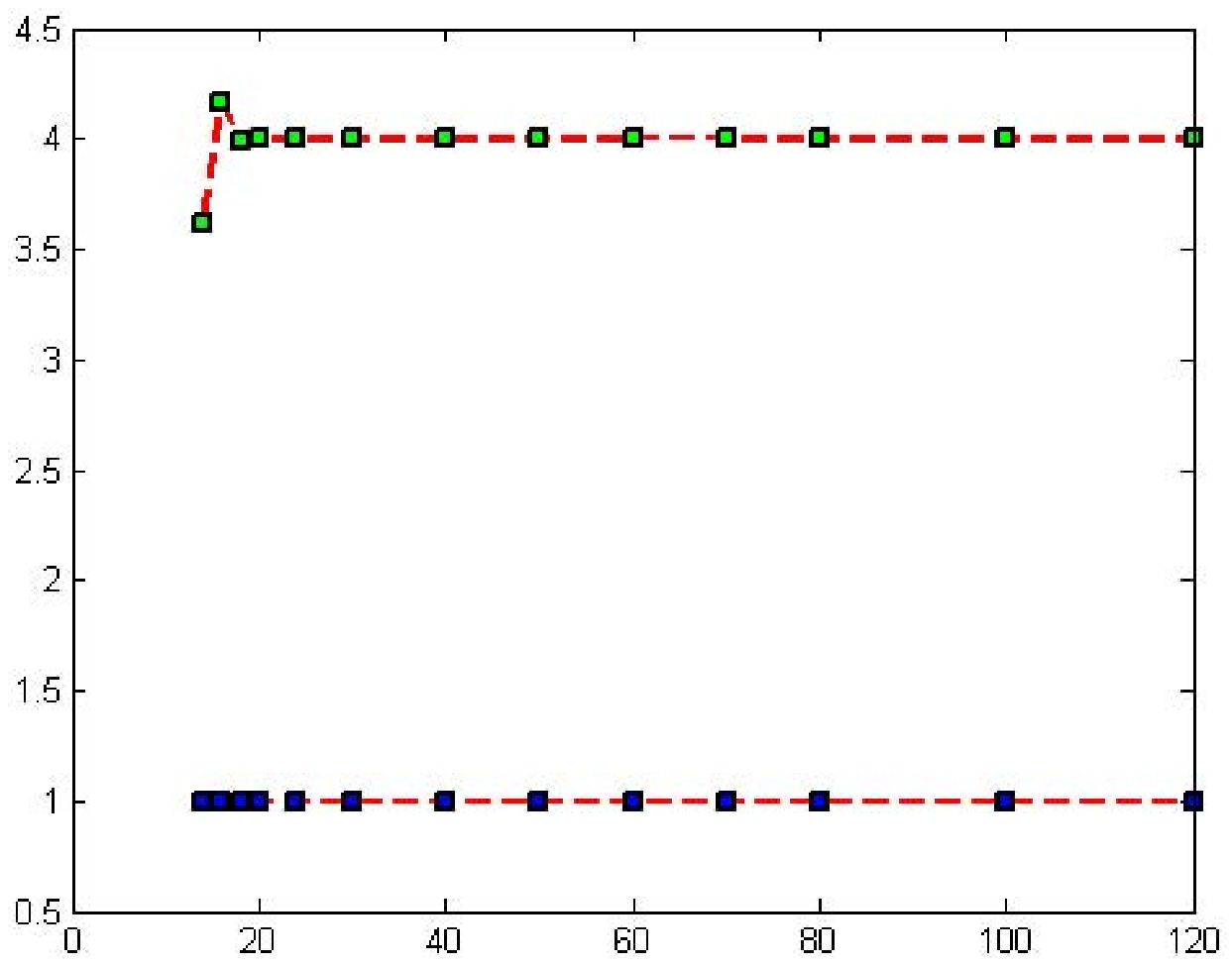}%
\caption{\small The approximate eigenvalues $\lambda_{1}$ and $\lambda_{2}$ from
Example \ref{Example1} computed using different number $N$ of formal powers.}%
\end{center}
\end{figure}

\begin{figure}
[ptb]
\begin{center}
\includegraphics[
natheight=4.379400in,
natwidth=5.862600in,
height=3.0044in,
width=4.0136in
]
{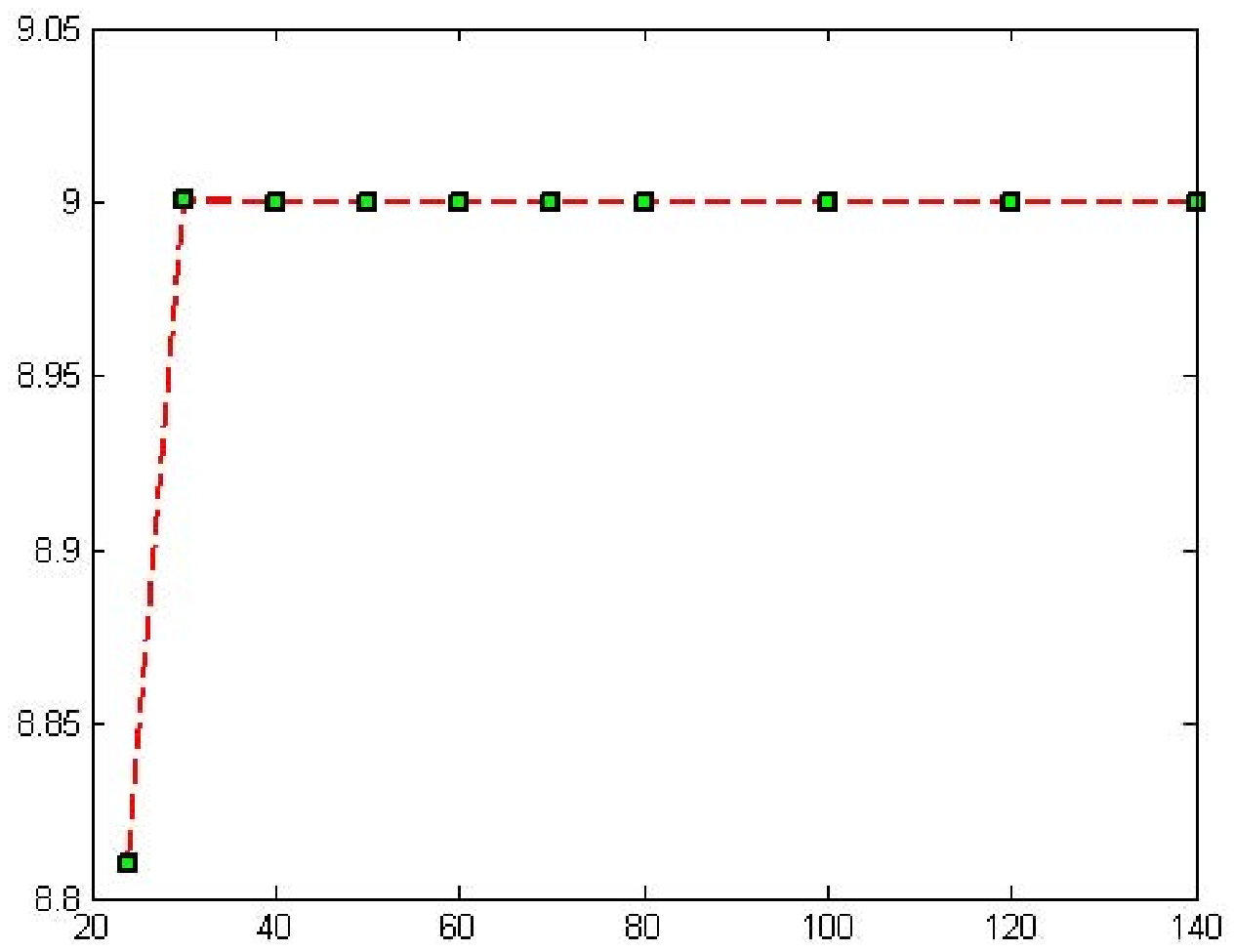}
\caption{\small The approximate values of $\lambda_{3}$ from Example \ref{Example1}
computed using different number $N$ of formal powers.}%
\end{center}
\end{figure}

\begin{figure}
[ptb]
\begin{center}
\includegraphics[
natheight=4.379400in,
natwidth=5.862600in,
height=3.0044in,
width=4.0136in
]%
{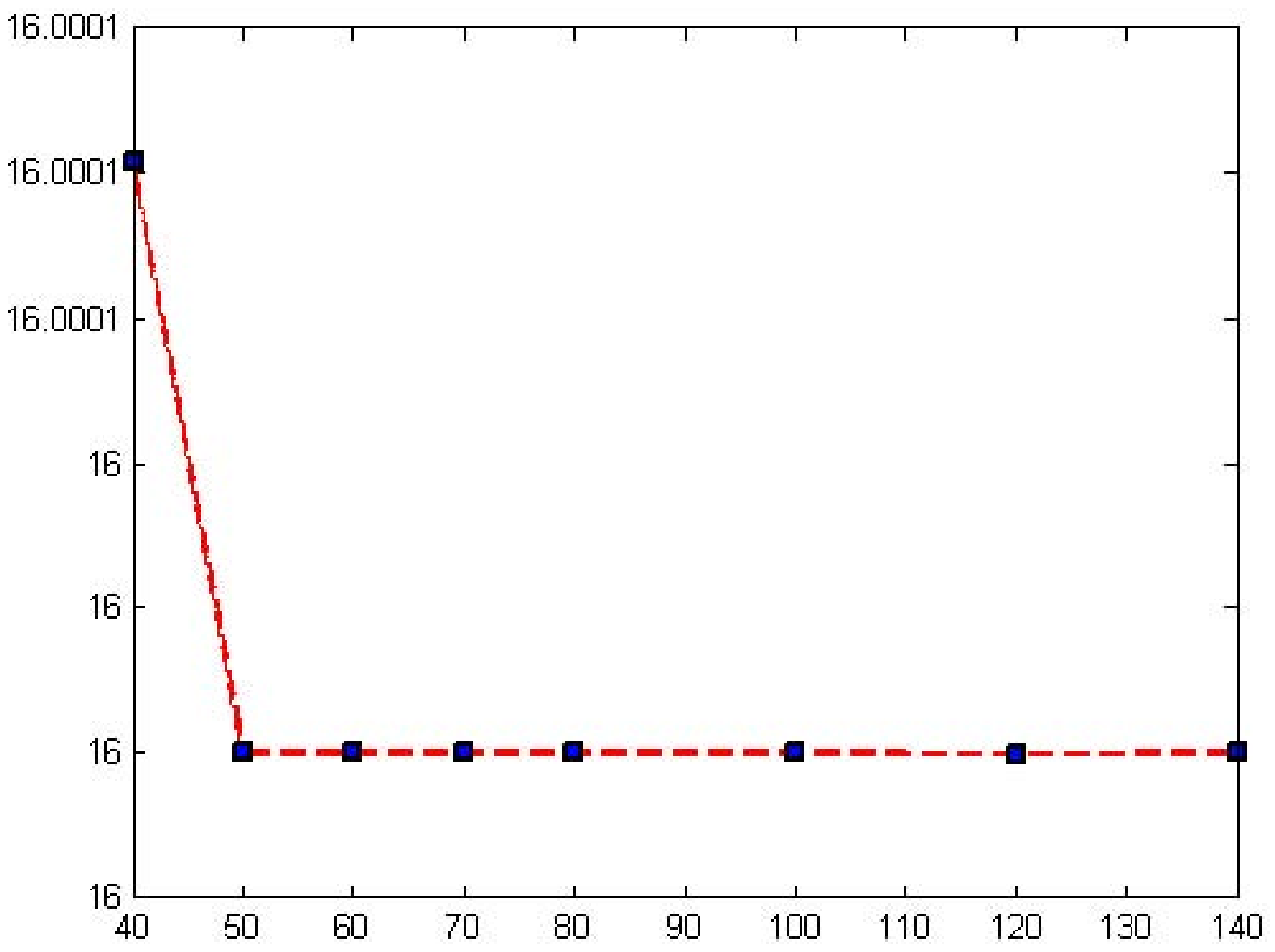}%
\caption{\small The approximate values of $\lambda_{4}$ from Example \ref{Example1}
computed using different number $N$ of formal powers.}%
\end{center}
\end{figure}

An attractive feature of the SPPS method is the possibility to easily plot the
characteristic relation. In Fig. 4 we show the absolute value of the
expression from the left-hand side of (\ref{eqSLparam}) as a function of the
complex variable $\lambda$ for the considered example. Its zeros coincide with
the eigenvalues of the problem. It is important to mention that such plot is
obtained in a fraction of a second. This is due to the fact that once the
required formal powers $X^{(n)}$ are computed (and this takes several seconds)
the calculation of the characteristic relation involves only simple algebraic
operations.
\begin{figure}
[ptb]
\begin{center}
\includegraphics[
height=2.45in,
width=4.6in
]%
{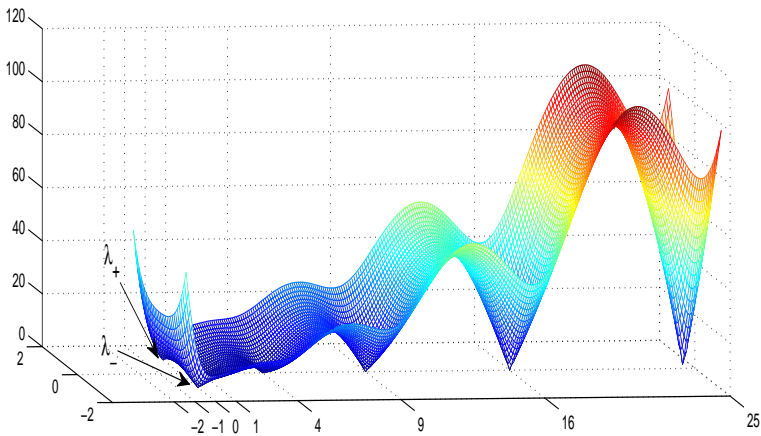}%
\caption{\small The absolute value of the expression from the left-hand side of
(\ref{eqSLparam}) as a function of the complex variable $\lambda$ for the
considered example. With the arrows we indicate the calculated complex
eigenvalues $\lambda_{\pm}$. The other zeros of the graph correspond to the
first five real eigenvalues of the problem.}%
\end{center}
\end{figure}
\end{ex}

Let us consider the one-dimensional Schr\"{o}dinger equation
\begin{equation}
Hu(x)=-u^{\prime\prime}(x)+Q(x)u(x)=\lambda u(x),\quad x\in\mathbb{R},
\label{Hu}%
\end{equation}
where
\begin{equation}
Q(x)=\begin{cases}
\alpha_{1}, & x<0,\\
q(x), & 0\leq x\leq h,\\
\alpha_{2}, & x>h,
\end{cases}  \label{Q}%
\end{equation}
$\alpha_{1}$ and $\alpha_{2}$ are complex constants and $q$ is a continuous
complex-valued function defined on the segment $[0,h]$. Thus, outside a finite
segment the potential $Q$ admits constant values, and at the end points of the
segment the potential may have discontinuities. We are looking for such values
of the spectral parameter $\lambda\in\mathbb{C}$ for which the Schr\"{o}dinger
equation possesses a solution $u$ belonging to the Sobolev space
$H^{2}(\mathbb{R)}$ which in the case of the potential of the form (\ref{Q})
means that we are looking for solutions exponentially decreasing at $\pm
\infty$. This eigenvalue problem is one of the central in quantum mechanics
for which $H$ is a self-adjoint operator in $L^{2}(\mathbb{R})$ with the
domain $H^{2}(\mathbb{R)}.$ It implies that $Q$ is a real-valued function. In
this case the operator $H$ has a continuous spectrum $\bigl[\min\left\{  \alpha
_{1},\alpha_{2}\right\}  ,+\infty\bigr)$ and a discrete spectrum located on the
set
\[
\Big[\min_{x\in[0,h]}q(x),\min\left\{  \alpha_{1},\alpha
_{2}\right\}  \Big). \label{interval}%
\]
Computation of energy levels of a quantum well described by the potential $Q$
is a problem of physics of semiconductor nanostructures (see, e.g.,
\cite{Harrison}). Other important models which reduce to the spectral problem
(\ref{Hu}) arise in studying the electromagnetic and acoustic wave propagation
in inhomogeneous waveguides (see for instance \cite{BC}, \cite{Ch}, \cite{FM},
\cite{WCC}, \cite{Breh}, \cite{OR}, \cite{MC}).

A characteristic equation for this spectral problem in terms of spectral
parameter power series was obtained in \cite{CKOR} (see also \cite{KKRosu})
where a simple numerical algorithm based on the approximation of the
characteristic equation was implemented and compared to other known numerical
techniques. Here we only give an example from \cite{CKOR}.

The usual approach to numerical solution of the considered eigenvalue problem
consists in applying the shooting method (see, e.g., \cite{Harrison}) which is
known to be unstable, relatively slow and to the difference of the SPPS
approach does not offer any explicit equation for determining eigenvalues and
eigenfunctions. In \cite{Hall} another method based on approximation of the
potential by square wells was proposed. It is limited to the case of symmetric
potentials. The approach based on the SPPS is completely different and does
not require any shooting procedure, approximation of the potential or
numerical differentiation. Derived from the exact characteristic equation its
approximation is considered, and in fact numerically the problem is reduced to
finding zeros of a polynomial $\sum_{k=0}^{N}a_{k}\mu^{k}$ in the interval
$\bigl[\min q(x),0\bigr)$, ($\mu^{2}=-\lambda$).

As an example, consider the potential $Q$ defined by the expression
$Q(x)=-\upsilon\operatorname{sech}^{2}x$, $x\in(-\infty,\infty)$. It is not of
a finite support, nevertheless its absolute value decreases rapidly when
$x\rightarrow\pm\infty$. The original problem is approximated by a problem
with a finite support potential $\widehat{Q}$ defined by the equality
\[
\widehat{Q}(x)=
\begin{cases}
0, & x<-a\\
-\upsilon\operatorname{sech}^{2}x, & -a\leq x\leq a\\
0, & x>a.
\end{cases}
\]
An attractive feature of the potential $Q$ is that its eigenvalues can be
calculated explicitly (see, e.g., \cite{Flugge}). In particular, for
$\upsilon=m(m+1)$ the eigenvalue $\lambda_{n}$ is given by the formula
$\lambda_{n}=-(m-n)^{2}$, $n=0,1,\ldots$. \bigskip

The results of application of the SPPS method for $\upsilon=12$ are given in
Table \ref{Table Sech} in comparison with the exact values and the results
from \cite{Hall}.
\begin{table}[h]
\begin{tabular}
[c]{|c|c|c|c|}\hline
$n$ & Exact values & Num.res. from \cite{Hall} & Num.res.
using SPPS ($N=180$)\\\hline
0 & $-9$ & $-9.094$ & $-8.999628656$\\\hline
1 & $-4$ & $-4.295$ & $-3.999998053$\\\hline
2 & $-1$ & $-0.885$ & $-0.999927816$\\\hline
\end{tabular}
\caption{Approximations of $\lambda_{n}$ of the
Hamiltonian \newline $H=-D^{2}-12\operatorname{sech}^{2}x$}
\label{Table Sech}%
\end{table}

The results obtained by means of SPPS are considerably more accurate, and as
was pointed out above the application of the SPPS method has much less restrictions.

\section{Transmutation operators}

\label{SubSectTransmSL} We slightly modify here the definition given by
Levitan \cite{LevitanInverse} adapting it to the purposes of the present work.
Let $E$ be a linear topological space and $E_{1}$ its linear subspace (not
necessarily closed). Let $A$ and $B$ be linear operators: $E_{1}\rightarrow
E$.

\begin{defn}
\label{DefTransmut} A linear invertible operator $T$ defined on the whole $E$
such that $E_{1}$ is invariant under the action of $T$ is called a
transmutation operator for the pair of operators $A$ and $B$ if it fulfills
the following two conditions.

\begin{enumerate}
\item Both the operator $T$ and its inverse $T^{-1}$ are continuous in $E$;
\item The following operator equality is valid
\begin{equation}
AT=TB \label{ATTB}%
\end{equation}
or which is the same
\[
A=TBT^{-1}.
\]
\end{enumerate}
\end{defn}

Very often in literature the transmutation operators are called the
transformation operators. Here we keep ourselves to the original term coined
by Delsarte and Lions \cite{DelsarteLions1956}. Our main interest concerns the
situation when $A=-\frac{d^{2}}{dx^{2}}+q(x)$, $B=-\frac{d^{2}}{dx^{2}}$,
\ and $q$ is a continuous complex-valued function. Hence for our purposes it
will be sufficient to consider the functional space $E=C[a,b]$ with the
topology of uniform convergence and its subspace $E_{1}$ consisting of
functions from $C^{2}\left[  a,b\right]  $. One of the possibilities to
introduce a transmutation operator on $E$ was considered by Lions
\cite{Lions57} and later on in other references (see, e.g., \cite{Marchenko}),
and consists in constructing a Volterra integral operator corresponding to a
midpoint of the segment of interest. As we begin with this transmutation
operator it is convenient to consider a symmetric segment $[-a,a]$ and hence
the functional space $E=C[-a,a]$. It is worth mentioning that other well known
ways to construct the transmutation operators (see, e.g.,
\cite{LevitanInverse}, \cite{Trimeche}) imply imposing initial conditions on
the functions and consequently lead to transmutation operators satisfying
(\ref{ATTB}) only on subclasses of $E_{1}$.

Thus, we consider the space $E=C[-a,a]$ and an operator of transmutation for
the defined above $A$ and $B$ can be realized in the form (see, e.g.,
\cite{LevitanInverse} and \cite{Marchenko}) of a Volterra integral operator
\begin{equation}
Tu(x)=u(x)+\int_{-x}^{x}K(x,t)u(t)dt\label{T}%
\end{equation}
where $K(x,t)=H\big(\frac{x+t}{2},\frac{x-t}{2}\big)$ and $H$ is the unique
solution of the Goursat problem%
\begin{equation}
\frac{\partial^{2}H(u,v)}{\partial u\,\partial v}%
=q(u+v)H(u,v),\label{GoursatH1}%
\end{equation}%
\begin{equation}
H(u,0)=\frac{1}{2}\int_{0}^{u}q(s)\,ds,\qquad H(0,v)=0.\label{GoursatH2}%
\end{equation}
If the potential $q$ is continuously differentiable, the kernel $K$ itself is
the solution of the Goursat problem
\[
\left(  \frac{\partial^{2}}{\partial x^{2}}-q(x)\right)  K(x,t)=\frac
{\partial^{2}}{\partial t^{2}}K(x,t),\label{Goursat1}%
\]%
\[
K(x,x)=\frac{1}{2}\int_{0}^{x}q(s)\,ds,\qquad K(x,-x)=0.\label{Goursat2}%
\]
If the potential $q$ is $n$ times continuously differentiable, the kernel
$K(x,t)$ is $n+1$ times continuously differentiable with respect to both
independent variables (see \cite{Marchenko}).

An important property of this transmutation operator consists in the way how
it maps solutions of the equation%
\begin{equation}
v^{\prime\prime}+\omega^{2}v=0 \label{SLomega1}%
\end{equation}
into solutions of the equation%
\begin{equation}
u^{\prime\prime}-q(x)u+\omega^{2}u=0 \label{SLomega2}%
\end{equation}
where $\omega$ is a complex number. Denote by $e_{0}(i\omega,x)$ the solution
of (\ref{SLomega2}) satisfying the initial conditions%
\[
e_{0}(i\omega,0)=1\qquad\text{and}\qquad e_{0}^{\prime}(i\omega,0)=i\omega.
\label{initcond}%
\]
The subindex ``$0$'' indicates that the initial conditions correspond to the
point $x=0$ and the letter ``$e$'' reminds us that the initial values coincide
with the initial values of the function $e^{i\omega x}$.

The transmutation operator (\ref{T}) maps $e^{i\omega x}$ into $e_{0}%
(i\omega,x)$,
\begin{equation}
e_{0}(i\omega,x)=T[e^{i\omega x}] \label{e0=Te}%
\end{equation}
(see \cite[Theorem 1.2.1]{Marchenko}).

Following \cite{Marchenko} we introduce the following notations%
\[
K_{c}(x,t;h)=h+K(x,t)+K(x,-t)+h\int_{t}^{x}\{K(x,\xi)-K(x,-\xi)\}d\xi
\]
where $h$ is a complex number, and
\[
K_{s}(x,t;\infty)=K(x,t)-K(x,-t).
\]

\begin{thm}
[\cite{Marchenko}]\label{TcTsMapsSolutions} Solutions $c(\omega,x;h)$ and
$s(\omega,x;\infty)$ of equation \eqref{SLomega2} satisfying the initial
conditions
\begin{equation}
c(\omega,0;h)=1,\qquad c_{x}^{\prime}(\omega,0;h)=h \label{ICcos}%
\end{equation}%
\begin{equation}
s(\omega,0;\infty)=0,\qquad s_{x}^{\prime}(\omega,0;\infty)=1 \label{ICsin}%
\end{equation}
can be represented in the form
\begin{equation}
c(\omega,x;h)=\cos\omega x+\int_{0}^{x}K_{c}(x,t;h)\cos\omega t\,dt
\label{c cos}%
\end{equation}
and
\begin{equation}
s(\omega,x;\infty)=\frac{\sin\omega x}{\omega}+\int_{0}^{x}K_{s}%
(x,t;\infty)\frac{\sin\omega t}{\omega}\,dt. \label{s sin}%
\end{equation}

\end{thm}

Denote by
\begin{equation}
T_{c}u(x)=u(x)+\int_{0}^{x}K_{c}(x,t;h)u(t)dt\label{Tc}%
\end{equation}
and%
\begin{equation}
T_{s}u(x)=u(x)+\int_{0}^{x}K_{s}(x,t;\infty)u(t)dt\label{Ts}%
\end{equation}
the corresponding integral operators. As was pointed out in \cite{CKT}, they
are not transmutations on the whole subspace $E_{1}$, they even do not map all
solutions of (\ref{SLomega1}) into solutions of (\ref{SLomega2}). For example,
as we show below
\[
\left(  -\frac{d^{2}}{dx^{2}}+q(x)\right)  T_{s}[1]\neq T_{s}\left[
-\frac{d^{2}}{dx^{2}}(1)\right]  =0
\]
when $q$ is constant.

\begin{ex}
\label{ExampleTransmut} Transmutation operator for operators $A:=\frac{d^{2}%
}{dx^{2}}+c$, $c$ is a constant, and $B:=\frac{d^{2}}{dx^{2}}$. According to
\eqref{GoursatH1} and \eqref{GoursatH2}, finding the kernel of transmutation
operator is equivalent to finding the function $H(s,t)=K(s+t,s-t)$ satisfying
the Goursat problem
\[
\frac{\partial^{2}H(s,t)}{\partial s\partial t}=-cH(s,t),\quad H(s,0)=-\frac
{cs}{2},\quad H(0,t)=0.
\]
The solution of this problem is given by \cite[(4.85)]{Garab}
\[
H(s,t)=-\frac{c}{2}\int_{0}^{s}J_{0}\bigl(2\sqrt{ct(s-\xi)}\bigr)\,d\xi
=-\frac{\sqrt{cst}J_{1}(2\sqrt{cst})}{2t},
\]
where $J_{0}$ and $J_{1}$ are Bessel functions of the first kind, and the
formula is valid even if the radicand is negative. Hence,
\begin{equation}
K(x,y)=H\left(  \frac{x+y}{2},\frac{x-y}{2}\right)  =-\frac{1}{2}\frac
{\sqrt{c(x^{2}-y^{2})}J_{1}\bigl(\sqrt{c(x^{2}-y^{2})}\bigr)}{x-y}.
\label{ExampleKernel}%
\end{equation}
From \eqref{ExampleKernel} we get the `sine' kernel
\[
K_{s}(x,t;\infty)=-\frac{t\sqrt{c(x^{2}-t^{2})}J_{1}\bigl(\sqrt{c(x^{2}%
-t^{2})}\bigr)}{x^{2}-t^{2}},\label{ExampleSineKernel}%
\]
and can check the above statement about the operator $T_{s}$,
\begin{gather*}
T_{s}[1](x)=1-\int_{0}^{x}\frac{t\sqrt{c(x^{2}-t^{2})}J_{1}\bigl(\sqrt
{c(x^{2}-t^{2})}\bigr)}{x^{2}-t^{2}}\,dt=J_{0}(x\sqrt{c}),\\
\left(  \frac{d^{2}}{dx^{2}}+c\right)  T_{s}[1]=\frac{\sqrt{c}J_{1}(x\sqrt
{c})}{x}\neq0.
\end{gather*}

\end{ex}

For the rest of this section suppose that $f$ is a solution of (\ref{SLhom})
fulfilling the condition of Theorem \ref{ThGenSolSturmLiouville} on a finite
segment $\left[  -a,a\right]  $. We normalize $f$ in such a way that $f(0)=1$
and let $f^{\prime}(0)=h$ where $h$ is some complex constant. Define the
system of functions $\left\{  \varphi_{k}\right\}  _{k=0}^{\infty}$ by this
function $f$ with the use of \eqref{Xn}, \eqref{Xtiln} and \eqref{phik}. The
system of functions $\left\{  \varphi_{k}\right\}  _{k=0}^{\infty}$ is related
to the transmutation operators $T_{c}$ (with the same parameter $h$ in the
kernel) and $T_{s}$ in a way that it is the union of functions which are the
result of acting of operator $T_{s}$ on the odd powers of independent variable
and of operator $T_{c}$ on the even powers of independent variable. The
following theorem holds, see \cite{CKT} for the details of the proof.

\begin{thm}
[\cite{CKT}]\label{Th Transmutation of Powers Tc and Ts} Let $q$ be a
continuous complex valued function of an independent real variable
$x\in[-a,a]$, and $f$ be a particular solution of \eqref{SLhom} such
that $f\in C^{2}\left(  -a,a\right)  $, $f\neq0$ on $[-a,a]$ and normalized as
$f(0)=1$. Let $\varphi_{k}$, $k\in\mathbb{N}_{0}:=\mathbb{N\cup}\left\{
0\right\}  $ are functions defined by \eqref{phik}. Then the following
equalities are valid%
\[
\varphi_{k}=T_{c}[x^{k}]\text{\quad when }k\in\mathbb{N}_{0}\text{ is even or
equal to zero}\label{Tc xk}%
\]
and
\[
\varphi_{k}=T_{s}[x^{k}]\text{\quad when }k\in\mathbb{N}\text{ is
odd.}\label{Ts xk}%
\]

\end{thm}

As for the transmutation operator $T$, it does not map all powers of the
independent variable into the functions $\varphi_{k}$. Instead, the following
theorem holds.

\begin{thm}
[\cite{CKT}]\label{Th Transmutation of Powers K}Under the conditions of
Theorem \ref{Th Transmutation of Powers Tc and Ts} the following equalities
are valid
\begin{equation}
\varphi_{k}=T[x^{k}]\text{\quad when }k\text{ is odd} \label{Txk_odd}%
\end{equation}
and
\begin{equation}
\varphi_{k}-\frac{h}{k+1}\varphi_{k+1}=T[x^{k}]\text{\quad when }%
k\in\mathbb{N}_{0}\text{ is even or equal to zero} \label{Txk_even1}%
\end{equation}
where by $h$ we denote $f^{\prime}(0)\in\mathbb{C}$.
\end{thm}

Taking into account the first of former relations the second can be written
also as follows%
\[
\varphi_{k}=T\left[  x^{k}+\frac{h}{k+1}x^{k+1}\right]  \text{\quad when }%
k\in\mathbb{N}_{0}\text{ is even or equal to zero.}\label{Txk_even2}%
\]

\begin{rem}
Let $f$ be the solution of (\ref{SLhom}) satisfying the initial conditions
\begin{equation}
f(0)=1,\quad\text{and}\quad f^{\prime}(0)=0. \label{initcond 1 0}%
\end{equation}
If it does not vanish on $[-a,a]$ then from Theorem
\ref{Th Transmutation of Powers K} we obtain that $\varphi_{k}=T[x^{k}]$ for
any $k\in\mathbb{N}_{0}$. In general, of course there is no guaranty that the
solution satisfying (\ref{initcond 1 0}) have no zero on $[-a,a]$. Hence the
operator $T$ transmutes the powers of $x$ into $\varphi_{k}$ whose
construction is based on the solution $f$ satisfying (\ref{initcond 1 0}) only
in some neighborhood of the origin. In the next section we show how to change
the operator $T$ so that the new operator map $x^{k}$ into $\varphi_{k}(x)$ on
the whole segment $[-a,a]$.
\end{rem}

Note that in Theorem \ref{Th Transmutation of Powers K} the operator $T$ does
not depend on the function $f$, so the right-hand sides of the equalities
\eqref{Txk_odd} and \eqref{Txk_even1} do not change with the change of $f$.
Consider two non-vanishing solutions $f$ and $g$ of \eqref{SLhom} normalized
as $f(0)=g(0)=1$ and let $\varphi_{k}^{f}$ and $\varphi_{k}^{g}$ be the
functions obtained from $f$ and $g$ respectively by means of \eqref{Xn},
\eqref{Xtiln} and \eqref{phik}. The relation between $\varphi_{k}^{f}$ and
$\varphi_{k}^{g}$ are given by the following proposition which may be easily
deduced from equalities \eqref{Txk_odd} and \eqref{Txk_even1}.

\begin{prop}
The following equalities hold
\[
\varphi_{k}^{f} =\varphi_{k}^{g}\text{\quad when }k\in\mathbb{N}\text{ is
odd,}%
\]
and
\[
\varphi_{k}^{f} =\varphi_{k}^{g}+\frac{h_{f}-h_{g}}{k+1}\varphi_{k+1}%
^{g}\text{\quad when }k\in\mathbb{N}_{0}\text{ is even,}%
\]
where $h_{f}=f^{\prime}(0)$ and $h_{g}=g^{\prime}(0)$.
\end{prop}

\section{A parametrized family of transmutation operators}

In \cite{CKT} a parametrized family of operators $\mathbf{T}_{h}$,
$h\in\mathbb{C}$ was introduced, given by the integral expression
\begin{equation}
\mathbf{T}_{h}u(x)=u(x)+\int_{-x}^{x}\mathbf{K}(x,t;h)u(t)dt \label{Tmain}%
\end{equation}
where
\begin{equation}
\mathbf{K}(x,t;h)=\frac{h}{2}+K(x,t)+\frac{h}{2}\int_{t}^{x}%
\big(K(x,s)-K(x,-s)\big)\,ds. \label{Kmain}%
\end{equation}
They are related to operators $T_{s}$ and $T_{c}$ (with the parameter $h$ in
the kernel of the latter operator) by
\begin{equation}
\mathbf{T}_{h}=T_{c}P_{e}+T_{s}P_{o}, \label{TcPe+TsPo}%
\end{equation}
where $P_{e}f(x)=\big(f(x)+f(-x)\big)/2$ and $P_{o}%
f(x)=\big(f(x)-f(-x)\big)/2$ are projectors onto even and odd functions,
respectively. In this section we show that the operators $\mathbf{T}_{h}$ are
transmutations, summarize their properties and in Theorem
\ref{Th Transmutation of Powers} we show how they act on powers of $x$.

Let us notice that $\mathbf{K}(x,t;0)=K(x,t)$ and that the expression
\begin{multline*}
\mathbf{K}(x,t;h)-\mathbf{K}(x,-t;h)=K(x,t)-K(x,-t)-\\ \frac{h}{2}\int_{-t}%
^{t}\left(  K(x,s)-K(x,-s)\right)  \,ds=K(x,t)-K(x,-t)
\end{multline*}
does not depend on $h$. Thus, it is possible to compute $\mathbf{K}(x,t;h)$
for any $h$ by a given $\mathbf{K}(x,t;h_{1})$ for some particular value
$h_{1}$. We formulate this result as the following statement.

\begin{thm}
[\cite{CKT}]\label{Kh_and_Kh1} The integral kernels $\mathbf{K}(x,t;h)$ and
$\mathbf{K}(x,t;h_{1})$ are related by the expression
\begin{equation}
\mathbf{K}(x,t;h)=\frac{h-h_{1}}{2}+\mathbf{K}(x,t;h_{1})+\frac{h-h_{1}}%
{2}\int_{t}^{x}\big( \mathbf{K}(x,s;h_{1})-\mathbf{K}(x,-s;h_{1})\big) \,ds.
\label{KmainChangeOfH}%
\end{equation}

\end{thm}

The operator $\mathbf{T}_{h}$ may be expressed in terms of another operator
$\mathbf{T}_{h_{1}}$ and in particular, in terms of the operator $T$. The
following proposition holds.

\begin{prop}
The operators $\mathbf{T}_{h_{1}}$ and $\mathbf{T}_{h_{2}}$ are related by the
expression
\begin{equation}
\label{Th2 via Th1}\mathbf{T}_{h_{2}}u=\mathbf{T}_{h_{1}}\bigg[u(x)+\frac
{h_{2}-h_{1}}2\int_{-x}^{x} u(t)\,dt\bigg]
\end{equation}
valid for any $u\in C[-a,a]$. In particular,
\begin{equation}
\label{Th via T}\mathbf{T}_{h}u=T\bigg[u(x)+\frac{h}2\int_{-x}^{x}
u(t)\,dt\bigg].
\end{equation}

\end{prop}

\begin{proof}
Using formulas \eqref{Tmain} and \eqref{Kmain} we obtain
\begin{multline*}
\mathbf{T}_{h}u=Tu+\frac{h}{2}\int_{-x}^{x}u(t)\,dt+\frac{h}{2}\int_{-x}%
^{x}u(t)\int_{t}^{x}K(x,s)\,ds\,dt-\\ \frac{h}{2}\int_{-x}^{x}u(t)\int_{-t}%
^{x}K(x,s)\,ds\,dt,
\end{multline*}
and after changing the order of integration in the last two integrals, we
have
\begin{multline*}
T_{h}u=Tu+\frac{h}{2}\int_{-x}^{x}u(t)\,dt+\frac{h}{2}\int_{-x}^{x}%
K(x,s)\int_{-x}^{s}u(t)\,dt\,ds-\\ \frac{h}{2}\int_{-x}^{x}K(x,s)\int_{-x}%
^{-s}u(t)\,dt\,ds=\displaybreak[2]\\
Tu+\frac{h}{2}\int_{-x}^{x}u(t)\,dt+\frac{h}{2}\int_{-x}^{x}K(x,s)\bigg[\int
_{-x}^{0}+\int_{0}^{s}u(t)\,dt\bigg]\,ds-\displaybreak[2]\\
-\frac{h}{2}\int_{-x}^{x}K(x,s)\bigg[\int_{-x}^{0}-\int_{-s}^{0}%
u(t)\,dt\bigg]\,ds=\displaybreak[2]\\
Tu+\frac{h}{2}\int_{-x}^{x}u(t)\,dt+\frac{h}{2}\int_{-x}^{x}K(x,s)\int
_{-s}^{s}u(t)\,dt\,ds=T\bigg[u(x)+\frac{h}{2}\int_{-x}^{x}u(t)\,dt\bigg].
\end{multline*}
Since $\int_{-x}^{x}\int_{-t}^{t}u(s)\,ds\,dt=0$ for any function $u\in
C[-a,a]$, we have from \eqref{Th via T} that
\begin{multline*}
\mathbf{T}_{h_{1}}\bigg[u(x)+\int_{-x}^{x}u(t)\,dt\bigg]=T\bigg[u(x)+\frac
{h_{2}-h_{1}}{2}\int_{-x}^{x}u(t)\,dt+\\
\frac{h_{1}}{2}\int_{-x}^{x}\bigg(u(t)+\frac{h_{2}-h_{1}}{2}\int_{-t}%
^{t}u(s)\,ds\bigg)\,dt\bigg]=T\bigg[u(x)+\frac{h_{2}}{2}\int_{-x}%
^{x}u(t)\,dt\bigg]=\mathbf{T}_{h_{2}}u.
\end{multline*}
\vskip-1em
\end{proof}

Using \eqref{ICcos}--\eqref{Ts} and \eqref{TcPe+TsPo} it is possible to check
how the operators $\mathbf{T}_{h}$ act on solutions of (\ref{SLomega1}).

\begin{prop}
[\cite{KrT2012}]\label{ThMapsSolutions} The operator $\mathbf{T}_{h}$ maps a
solution $v$ of an equation $v^{\prime\prime}+\omega^{2}v=0$, where $\omega$
is a complex number, into a solution $u$ of the equation $u^{\prime\prime
}-q(x)u+\omega^{2}u=0$ with the following correspondence of the initial
values
\begin{equation}
u(0)=v(0),\qquad u^{\prime}(0)=v^{\prime}(0)+hv(0). \label{Th Initial Values}%
\end{equation}

\end{prop}

\begin{rem}
Formulas \eqref{Th Initial Values} are valid for any function $v\in
C^{1}[-a,a]$.
\end{rem}

We know that the integral kernel of the transmutation operator $T$ is related
to the solution of the Goursat problem \eqref{GoursatH1}--\eqref{GoursatH2}. A
similar result holds for the operators $\mathbf{T}_{h}$.

\begin{thm}
[\cite{KrT2012}]\label{TmainGoursat} In order for the function $K(x,t;h)$ to
be the kernel of a transmutation operator acting as described in Proposition
\ref{ThMapsSolutions}, it is necessary and sufficient that
$H(u,v;h):=K(u+v,u-v;h)$ be a solution of the Goursat problem
\[
\frac{\partial^{2}H(u,v;h)}{\partial u\,\partial v}%
=q(u+v)H(u,v;h),\label{GoursatTh1}%
\]%
\[
H(u,0;h)=\frac{h}{2}+\frac{1}{2}\int_{0}^{u}q(s)\,ds,\qquad H(0,v;h)=\frac
{h}{2}.\label{GoursatTh2}%
\]
If the potential $q$ is continuously differentiable, the function $K(x,t;h)$
itself must be the solution of the Goursat problem
\begin{equation}
\left(  \frac{\partial^{2}}{\partial x^{2}}-q(x)\right)  K(x,t;h)=\frac
{\partial^{2}}{\partial t^{2}}K(x,t;h), \label{GoursatTk1}%
\end{equation}%
\begin{equation}
K(x,x;h)=\frac{h}{2}+\frac{1}{2}\int_{0}^{x}q(s)\,ds,\qquad K(x,-x;h)=\frac
{h}{2}. \label{GoursatTk2}%
\end{equation}

\end{thm}

Under some additional requirements on the potential $q$ the operators
$\mathbf{T}_{h}$ are transmutations in the sense of Definition
\ref{DefTransmut}. The following theorem generalizes the results obtained in
\cite{KrT2012}.

\begin{thm}
\label{Th Transmutation} Suppose the potential $q$ satisfies either of the
following two conditions.

\begin{itemize}
\item $q\in C^{1}[-a,a]$;

\item $q\in C[-a,a]$ and there exists a particular complex-valued solution $g$
of \eqref{SLhom} non-vanishing on $[-a,a]$.
\end{itemize}

Then the operator $\mathbf{T}_{h}$ given by \eqref{Tmain} satisfies the
equality
\begin{equation}
\left(  -\frac{d^{2}}{dx^{2}}+q(x)\right)  \mathbf{T}_{h}[u]=\mathbf{T}%
_{h}\left[  -\frac{d^{2}}{dx^{2}}(u)\right]  \label{ThTransm}%
\end{equation}
for any $u\in C^{2}[-a,a]$.
\end{thm}

\begin{proof}
In \cite{KrT2012} the theorem was proved for the case $q\in C^{1}[-a,a]$ and
for the case when the particular solution $g$ from the statement satisfies
conditions $g(0)=1$ and $g^{\prime}(0)=h$.

We may normalize the particular solution $g$ as $g(0)=1$. Suppose that
$g^{\prime}(0)=h_{1}$. We know already that \eqref{ThTransm} holds for the
operator $\mathbf{T}_{h_{1}}$. To finish the proof, we use \eqref{Th2 via Th1}
and obtain
\begin{multline*}
\left(  -\frac{d^{2}}{dx^{2}}+q(x)\right)  \mathbf{T}_{h}[u]=\left(
-\frac{d^{2}}{dx^{2}}+q(x)\right)  \mathbf{T}_{h_{1}}\bigg[u(x)+\frac{h-h_{1}%
}2\int_{-x}^{x} u(t)\,dt\bigg]\\
=-\mathbf{T}_{h_{1}}\bigg[u^{\prime\prime}(x)+\frac{h-h_{1}}2\frac{d^{2}%
}{dx^{2}}\int_{-x}^{x} u(t)\,dt\bigg]=\\
-\mathbf{T}_{h_{1}}\bigg[u^{\prime
\prime}(x)+\frac{h-h_{1}}2\int_{-x}^{x} u^{\prime\prime}%
(t)\,dt\bigg]=\mathbf{T}_{h}\left[  -\frac{d^{2}}{dx^{2}}(u)\right]  .
\end{multline*}
\vskip-1em
\end{proof}

\begin{rem}
As was pointed out in Remark \ref{RemarkNonVanish}, in the regular case the
non-vanishing solution $g$ of \eqref{SLhom} exists due to the alternation of
zeroes of two linearly independent solutions. Of course, it would be
interesting to prove that the operators $\mathbf{T}_{h}$ are transmutations in
the general case of complex-valued potentials $q\in C[-a,a]$ without any
additional assumption.
\end{rem}

Suppose now that a function $f$ is a solution of \eqref{SLhom}, non-vanishing
on $[-a,a]$ and normalized as $f(0)=1$. Let $h:=f^{\prime}(0)$ be some complex
constant. Define as before the system of functions $\{\varphi_{k}%
\}_{k=0}^{\infty}$ by this function $f$ and by \eqref{phik}. The following
theorem states that the operator $\mathbf{T}_{h}$ transmutes powers of $x$
into the functions $\varphi_{k}$.

\begin{thm}
[\cite{CKT}]\label{Th Transmutation of Powers} Let $q$ be a continuous complex
valued function of an independent real variable $x\in[-a,a]$, and $f$ be
a particular solution of \eqref{SLhom} such that $f\in C^{2}(-a,a)$ together
with $1/f$ are bounded on $[-a,a]$ and normalized as $f(0)=1$, and let
$h:=f^{\prime}(0)$, where $h$ is a complex number. Then the operator
\eqref{Tmain} with the kernel defined by \eqref{Kmain} transforms $x^{k}$ into
$\varphi_{k}(x)$ for any $k\in\mathbb{N}_{0}$.
\end{thm}

Thus, we clarified that the system of functions $\{\varphi_{k}\}$ may be
obtained as the result of the Volterra integral operator acting on powers of
the independent variable. As was mentioned before, this offers an algorithm
for transmuting functions in the situation when $\mathbf{K}(x,t;h)$ is
unknown. Moreover, properties of the Volterra integral operator such as
boundedness and bounded invertibility in many functional spaces gives us a
tool to prove the completeness of the system of function $\{\varphi_{k}\}$ in
various situations.

\begin{ex}
\label{modelExample} Consider a function $k(x,t)=\frac{t-1}{2(x+1)}$ (later,
in Example \ref{DarbouxExample1} it is explained how it can be obtained). We
have
\begin{equation*}
(\partial_{x}^{2}-\partial_{t}^{2})k(x,t)=\frac{t-1}{(x+1)^{3}}=\frac
{2}{(x+1)^{2}}\cdot\frac{t-1}{2(x+1)},
\end{equation*}
$k(x,-x)=\frac{-x-1}{2(x+1)}=-\frac{1}{2}$ and $k(x,x)=\frac
{x-1}{2(x+1)}=-\frac{1}{2}+\frac{1}{2}\int_{0}^{x}\frac{2}{(s+1)^{2}}\,ds$,
thus the function $k(x,t)$ satisfies the Goursat problem
\eqref{GoursatTk1}--\eqref{GoursatTk2} with \linebreak $q(x)=2/(x+1)^{2}$ and $h=-1$ and
by Theorem \ref{TmainGoursat} is the kernel of the transmutation operator
$\mathbf{T}_{-1}$.

Consider the function $f=\mathbf{T}_{-1}[1]=\frac{1}{x+1}$ as a solution of
\eqref{SLhom} such that $f(0)=1$ and $f^{\prime}(0)=h=-1$, nonvanishing on any
$[-a,a]\subset(-1,1)$. The first 3 functions $\varphi_{k}$ are given by
\[
\varphi_{0}=f=\frac{1}{x+1},\quad\varphi_{1}=\frac{x^{3}+3x^{2}+3x}%
{3(x+1)},\quad\varphi_{2}=\frac{2x^{3}+3x^{2}}{3(x+1)}.
\]
It can be easily checked that
\begin{align*}
\mathbf{T}_{-1}x  &  =x+\int_{-x}^{x}\frac{(t-1)t}{2(x+1)}\,dt=\frac
{x^{3}+3x^{2}+3x}{3(x+1)} = \varphi_{1},\\
\mathbf{T}_{-1}x^{2}  &  =x^{2}+\int_{-x}^{x}\frac{(t-1)t^{2}}{2(x+1)}%
\,dt=\frac{2x^{3}+3x^{2}}{3(x+1)} = \varphi_{2}.
\end{align*}
We can calculate the kernel $K$ of the original operator $T$ by
\eqref{KmainChangeOfH}, it is given by
\[
K(x,t)=\frac{2x+2t+x^{2}-t^{2}}{4(x+1)}
\]
and we can check that $T[x]=\varphi_{1}$ and $T[1]=\frac{x^{3}+3x^{2}%
+3x+3}{3(x+1)}=\varphi_{0}+\varphi_{1}$ in accordance with Theorem
\ref{Th Transmutation of Powers K}.
\end{ex}

\section{Transmutation operators and Darboux transformed
equations\label{Sect Transmutations for Darboux}}

To construct the system of functions $\{\varphi_{k}\}_{k=0}^{\infty}$ we use
the half of the functions $\big\{X^{(k)},\widetilde{X}^{(k)}\big\}_{k=0}^{\infty}$.
What about the second half? Note that starting with the function $1/f$ we
obtain the same system of functions $\big\{X^{(k)},\widetilde{X}^{(k)}%
\big\}_{k=0}^{\infty}$ with the only change that $X_{f}^{(k)}$ becomes
$\widetilde{X}_{1/f}^{(k)}$ and $\widetilde{X}_{f}^{(k)}$ becomes
$X_{1/f}^{(k)}$. Hence the \textquotedblleft second half\textquotedblright\ of
the functions $\big\{  X^{(k)},\widetilde{X}^{(k)}\big\}  _{k=0}^{\infty}$
from \eqref{phik} is used. The function $1/f$ is continuous complex-valued and
non-vanishing, and is a solution of the equation $u^{\prime\prime}-q_{2}u=0$,
where $q_{2}=2\left(  \frac{f^{\prime}}{f}\right)  ^{2}-q$. The last equation
is known as the Darboux transformation of the original equation. The Darboux
transformation is closely related to the factorization of the Schr\"{o}dinger
equation, and nowadays it is used in dozens of works, see, e.g., \cite{Cies,
GHZh, Matveev, RS} in connection with solitons and integrable systems, e.g.,
\cite{BS, HV, NPS, PPS} and the review \cite{Rosu} of applications to quantum mechanics.

For the convenience denote the potential of the original equation by $q_{1}$
and the corresponding Sturm-Liouville operator by $A_{1}:=\frac{d^{2}}{dx^{2}%
}-q_{1}(x)$. Suppose a solution $f$ of the equation $A_{1}f=0$ is given such
that $f(x)\neq0,\ x\in[-a,a]$, it is normalized as $f(0)=1$ and
$h:=f^{\prime}(0)$ is some complex number. Denote the Darboux-transformed
operator by $A_{2}:=\frac{d^{2}}{dx^{2}}-q_{2}(x)$, where $q_{2}%
(x)=2\big(\frac{f^{\prime}(x)}{f(x)}\big)^{2}-q_{1}(x)$.

From the previous section we know that there exists a transmutation operator
$\mathbf{T}_{1;h}$ for the original equation (\ref{SLlambda}) with the
potential $q_{1}$ and such that
\begin{equation}
\mathbf{T}_{1;h}x^{k}=\varphi_{k},\quad k\in\mathbb{N}_{0}.\label{T1x^k}%
\end{equation}
The subindex \textquotedblleft$1$\textquotedblright\ in the notation
$\mathbf{T}_{1;h}$\ indicates that the transmutation operator corresponds to
$A_{1}$.

Similarly, there exists a transmutation operator $\mathbf{T}_{2;-h}$ for the
Darboux-transformed operator $A_{2}$ such that
\begin{equation}
\mathbf{T}_{2;-h}x^{k}=\psi_{k},\quad k\in\mathbb{N}_{0},\label{T2x^k}%
\end{equation}
where the family of functions $\{\psi_{k}\}_{k=0}^{\infty}$ is defined by
(\ref{psik}).

It is interesting to obtain some relations between the operators
$\mathbf{T}_{1;h}$ and $\mathbf{T}_{2;-h}$ and between their integral kernels
$\mathbf{K}_{1}$ and $\mathbf{K}_{2}$. In this section we explain how to
construct the integral kernel $\mathbf{K}_{2}$ by the known integral kernel
$\mathbf{K}_{1}$ and show that the operators $\mathbf{T}_{1;h}$ and
$\mathbf{T}_{2;-h}$ satisfy certain commutation relations with the operator of differentiation.


We remind some well known facts about the Darboux transformation. First, $1/f$
is a solution of $A_{2}u=0$. Second, it is closely related to the
factorization of Sturm-Liouville and one-dimensional Schr\"{o}dinger
operators. Namely, we have
\begin{align}
A_{1}=\frac{d^{2}}{dx^{2}}-q_{1}(x)  &  =\Big(\partial_{x}+\frac{f^{\prime}%
}{f}\Big)\Big(\partial_{x}-\frac{f^{\prime}}{f}\Big)=\frac{1}{f}\partial
_{x}f^{2}\partial_{x}\frac{1}{f}\cdot,\label{A1factor}\\
A_{2}=\frac{d^{2}}{dx^{2}}-q_{2}(x)  &  =\Big(\partial_{x}-\frac{f^{\prime}%
}{f}\Big)\Big(\partial_{x}+\frac{f^{\prime}}{f}\Big)=f\partial_{x}\frac
{1}{f^{2}}\partial_{x}f\cdot. \label{A2factor}%
\end{align}
Suppose that $u$ is a solution of the equation $A_{1}u=\omega u$ for some
$\omega\in\mathbb{C}$. Then the function $v=\big(\partial_{x}-\frac{f^{\prime
}}{f}\big)u=\big(f\partial_{x}\frac{1}{f}\big)u$ is a solution of the equation
$A_{2}v=\omega v$, and vice versa, given a solution $v$ of $A_{2}v=\omega v$,
the function $u=\big(\partial_{x}+\frac{f^{\prime}}{f}\big)v=\big(\frac{1}%
{f}\partial_{x}f\big)v$ is a solution of $A_{1}u=\omega u$.

Suppose that the operator $\mathbf{T}_{1}:=\mathbf{T}_{1;h}$ which transmutes
the operator $A_{1}$ into the operator $B=d^{2}/dx^{2}$ and the powers $x^{k}$
into the functions $\varphi_{k}$ is known in the sense that its kernel
$\mathbf{K}_{1}(x,t;h)$ is given. As before $h=f^{\prime}(0)$. Then the
function $1/f$ is the non-vanishing solution of the equation $A_{2}u=0$
satisfying $1/f(0)=1$ and $(1/f)^{\prime}(0)=-h$. Hence we are looking for the
operator $\mathbf{T}_{2}:=\mathbf{T}_{2;-h}$ transmuting the operator $A_{2}$
into the operator $B$ and the powers $x^{k}$ into the functions $\psi_{k}$.

Let us explain the idea for obtaining the operator $\mathbf{T}_{2}$. We want
to find an operator transforming solutions of the equation $Bu+\omega^{2}u=0$
into solutions of the equation $A_{2}u+\omega^{2}u=0$, see the first diagram
below. Starting with a solution $\sigma$ of the equation $(\partial_{x}%
^{2}+\omega^{2})\sigma=0$, by application of $\mathbf{T}_{1}$ we get a
solution of $(A_{1}+\omega^{2})u=0$, and the expression $\big(f\partial
_{x}\frac{1}{f}\big)\mathbf{T}_{1}\sigma$ is a solution of $(A_{2}+\omega
^{2})v=0$. But the operator $\big(f\partial_{x}\frac{1}{f}\big)\mathbf{T}_{1}$
is unbounded and hence cannot coincide with the operator $\mathbf{T}_{2}$. In
order to find the required bounded operator we may consider the second copy of
the equation $(\partial_{x}^{2}+\omega^{2})u=0$, which is a result of the
Darboux transformation applied to $(\partial_{x}^{2}+\omega^{2})\sigma=0$ with
respect to the particular solution $g\equiv1$ and construct the operator
$\mathbf{T}_{2}$ by making the second diagram commutative. In order to obtain
a bounded operator $\mathbf{T}_{2}$, instead of using $f\partial_{x}\frac
{1}{f}$ for the last step, we will use the inverse of $\frac{1}{f}\partial
_{x}f$, i.e. $\frac{1}{f}\big(\int_{0}^{x}f(s)\cdot\,ds+C\big)$.
\[
\xymatrix@R+1pt@C+12pt@M+1pt{
\partial_{x}^{2}+\omega^{2} \ar[r]^(.45){\mathbf{T}_1} \ar[dr]_(.45){\mathbf{T}_2} & \partial_{x}^{2}-q_{1}+\omega^{2} \ar[d]^{f\partial_{x}\frac{1}{f}}\\
& \partial_{x}^{2}-q_{2}+\omega^{2}
}\qquad\quad
\xymatrix@R+1pt@C+12pt@M+1pt{
\partial_{x}^{2}+\omega^{2} \ar[r]^(.45){\mathbf{T}_1} & \partial_{x}^{2}-q_{1}+\omega^{2} \ar@<1ex>[d]^{\frac{1}{f}(\int
f\cdot+C)}\\
\partial_{x}^{2}+\omega^{2} \ar[u]^{\partial_{x}} \ar[r]^(.45){\mathbf{T}_2} & \partial_{x}^{2}-q_{2}+\omega^{2} \ar@<1ex>[u]^{\frac{1}{f}
\partial_{x}f}
}
\]

That explains how to obtain the following theorem.
\begin{thm}
[\cite{KrT2012}]\label{Th T2Integral} The operator $T_{2}$, acting on
solutions $u$ of equations $(\partial_{x}^{2}+\omega^{2})u=0,\ \omega
\in\mathbb{C}$ by the rule
\begin{equation}
T_{2}[u](x)=\frac{1}{f(x)}\bigg(\int_{0}^{x}f(\eta)\mathbf{T}_{1}[u^{\prime
}](\eta)\,d\eta+u(0)\bigg) \label{T2sol}%
\end{equation}
coincides with the transmutation operator $\mathbf{T}_{2;-h}$.
\end{thm}

Now we show that the operator $T_{2}$ can be written as a Volterra integral
operator and, as a consequence, extended by continuity to a wider class of
functions. To obtain simpler expression for the integral kernel $\mathbf{K}%
_{2}(x,t;-h)$ we have to suppose that the original integral kernel
$\mathbf{K}_{1}(x,t;h)$ is known in the larger domain than required by
definition \eqref{Tmain}. Namely, suppose that the function $\mathbf{K}%
_{1}(x,t;h)$ is known and is continuously differentiable in the domain
$\bar{\Pi}:\ -a\leq x\leq a,-a\leq t\leq a$. We refer the reader to
\cite{KrT2012} for further details.

\begin{thm}
[\cite{KrT2012}]\label{Th T2Volterra} The operator $T_{2}$ admits a
representation as the Volterra integral operator
\begin{equation}
T_{2}[u](x)=u(x)+\int_{-x}^{x}\mathbf{K}_{2}(x,t;-h)u(t)\,dt, \label{T2}%
\end{equation}
with the kernel
\begin{equation}
\mathbf{K}_{2}(x,t;-h)=-\frac{1}{f(x)}\bigg(\int_{-t}^{x}\partial
_{t}\mathbf{K}_{1}(s,t;h)f(s)\,ds+\frac{h}{2}f(-t)\bigg). \label{K2}%
\end{equation}
Such representation is valid for any function $u\in C^{1}[-a,a]$.
\end{thm}

By Theorems \ref{Th T2Integral} and \ref{Th T2Volterra} the Volterra operators
$T_{2}$ and $\mathbf{T}_{2}$ coincide on the set of finite linear combinations
of solutions of the equations $(\partial_{x}^{2}+\omega^{2})u=0,\ \omega
\in\mathbb{C}$. Since this set is dense in $C[-a,a]$, by continuity of $T_{2}$
and $\mathbf{T}_{2}$ we obtain the following corollaries.

\begin{cor}
[\cite{KrT2012}]The operator $T_{2}$ given by \eqref{T2} with the kernel
\eqref{K2} coincides with $\mathbf{T}_{2}$ on $C[-a,a]$.
\end{cor}

\begin{cor}
[\cite{KrT2012}]The operator $T_{2}$ given by \eqref{T2sol} coincides with
$\mathbf{T}_{2}$ on $C^{1}[-a,a]$.
\end{cor}

Operator $A_{1}$ is the Darboux transformation of the operator $A_{2}$ with
respect to the solution $1/f$, hence we may obtain another relation between
the operators $\mathbf{T}_{1}$ and $\mathbf{T}_{2}$.

\begin{cor}
[\cite{KrT2012}]For any function $u\in C^{1}[-a,a]$ the equality
\begin{equation}
\mathbf{T}_{1}[u](x)=f(x)\bigg(\int_{0}^{x}\frac{1}{f(\eta)}\mathbf{T}%
_{2}[u^{\prime}](\eta)\,d\eta+u(0)\bigg) \label{T1sol}%
\end{equation}
is valid.
\end{cor}

From the second commutative diagram at the beginning of this subsection we may
deduce some commutation relations between the operators $\mathbf{T}_{1}$,
$\mathbf{T}_{2}$ and $d/dx$. The proof immediately follows from \eqref{T2sol}
and \eqref{T1sol}.

\begin{cor}
[\cite{KrT2012}]\label{Cor Commutation Relations}The following operator
equalities hold on $C^{1}[-a,a]$:
\begin{align}
\partial_{x}f\mathbf{T}_{2}  &  =f\mathbf{T}_{1}\partial_{x}\label{CommutT1dx}%
\\
\partial_{x}\frac{1}{f}\mathbf{T}_{1}  &  =\frac{1}{f}\mathbf{T}_{2}%
\partial_{x}. \label{CommutT2dx}%
\end{align}

\end{cor}

In \cite{KMoT} the following notion of generalized derivatives was introduced.
Consider a function $g$ assuming that both $f$ and $g$ possess the derivatives
of all orders up to the order $n$ on the segment $[-a,a]$. Then in $[-a,a]$
the following generalized derivatives are defined
\begin{align*}
\gamma_{0}(g)(x)  &  =g(x),\\
\gamma_{k}(g)(x)  &  =\big(f^{2}(x)\big)^{(-1)^{k-1}}\big(\gamma
_{k-1}(g)\big)^{\prime}(x)
\end{align*}
for $k=1,2,\ldots,n$.

Let a function $u$ be defined by the equality
\[
g=\frac{1}{f}\mathbf{T}_{1}u,
\]
and assume that $u\in C^{n}[-a,a]$. Note that below we do not necessarily
require that the functions $f$ and $g$ be from $C^{n}[-a,a]$. With the use of
\eqref{CommutT1dx} and \eqref{CommutT2dx} we have
\begin{align*}
\gamma_{1}(g)  &  =f^{2}\cdot\Big(\frac{1}{f}\mathbf{T}_{1}u\Big)^{\prime
}=f^{2}\cdot\frac{1}{f}\mathbf{T}_{2}u^{\prime}=f\mathbf{T}_{2}u^{\prime},\\
\gamma_{2}(g)  &  =\frac{1}{f^{2}}\cdot\Big(f\mathbf{T}_{2}u^{\prime
}\Big)^{\prime}=\frac{1}{f^{2}}\cdot f\mathbf{T}_{1}u^{\prime\prime}=\frac
{1}{f}\mathbf{T}_{1}u^{\prime\prime}.
\end{align*}
By induction we obtain the following corollary.

\begin{cor}
[\cite{KrT2012}]\label{GeneralDerivTransm} Let $u\in C^{n}[-a,a]$ and
$g=\frac{1}{f}\mathbf{T}_{1}u$. Then
\[
\gamma_{k}(g)=f\mathbf{T}_{2}u^{(k)}\qquad\text{if\ }k\text{\ is odd,}\ k\leq
n,
\]
and
\[
\gamma_{k}(g)=\frac{1}{f}\mathbf{T}_{1}u^{(k)}\qquad\text{if\ }k\text{\ is
even,}\ k\leq n.
\]

\end{cor}

\begin{ex}
\label{DarbouxExample1} We start with the operator $A_{0}=d^{2}/dx^{2}$. We
have to pick up such a solution $f$ of the equation $A_{0}f=0$ that
$f^{\prime}/f\neq0$. This is in order to obtain an operator $A_{1}\neq A_{0}$
as a result of the Darboux transformation of $A_{0}$. For such solution $f$
consider, e.g., $f_{0}(x)=x+1$. Both $f_{0}$ and $1/f_{0}$ are bounded on any
segment $[-a,a]\subset(-1;1)$ and the Darboux transformed operator has the
form $A_{1}=\frac{d^{2}}{dx^{2}}-\frac{2}{(x+1)^{2}}$.

The transmutation operator $T$ for $A_{0}$ is obviously an identity operator
and $K_{0}(x,t;0)=0$. Since $f_{0}^{\prime}(0)=1$, we look for the
parametrized operator $\mathbf{T}_{0;1}$. Its kernel is given by
\eqref{KmainChangeOfH}: $\mathbf{K}_{0}(x,t;1)=1/2$. From Theorem
\ref{Th T2Volterra} we obtain the transmutation kernel for the operator
$A_{1}$
\begin{equation}
\mathbf{K}_{1}(x,t;-1)=-\frac{1}{x+1}\cdot\frac{1-t}{2}=\frac{t-1}{2(x+1)},
\label{ExampleK1}%
\end{equation}
the kernel from Example \ref{modelExample}.

To obtain a less trivial example consider again the operator $A_{1}%
=\frac{d^{2}}{dx^{2}}-\frac{2}{(x+1)^{2}}$ and the function $f_{1}%
(x)=(x+1)^{2}$ as a solution of $A_{1}f=0$. Since $h=f_{1}^{\prime}(0)=2$, we
compute $\mathbf{K}_{1}(x,t;2)$ from \eqref{ExampleK1} using
\eqref{KmainChangeOfH}
\[
\mathbf{K}_{1}(x,t;2)=\frac{3x^{2}+6x+4-3t^{2}+2t}{4(x+1)}.
\]
The Darboux transformation of the operator $A_{1}$ with respect to the
solution $f_{1}$ is the operator $A_{2}=\frac{d^{2}}{dx^{2}}-\frac
{6}{(x+1)^{2}}$ and by Theorem \ref{Th T2Volterra} the transmutation operator
$\mathbf{T}_{2;-2}$ for $A_{2}$ is given by the Volterra integral operator
\eqref{Tmain} with the kernel
\begin{multline*}
\mathbf{K}_{2}(x,t;-2)=-\frac{1}{(x+1)^{2}}\bigg(\int_{-t}^{x}\frac
{-3t+1}{2(s+1)}(s+1)^{2}\,ds+(1-t)^{2}\bigg)=\\ \frac{(3t-1)(x+1)^{2}%
-3(t-1)^{2}(t+1)}{4(x+1)^{2}}.
\end{multline*}

This procedure may be continued iteratively. Consider the operators
\[
A_{n}:=\frac{d^{2}}{dx^{2}}-\frac{n(n+1)}{(x+1)^{2}}.
\]
The function $f_{n}(x)=(x+1)^{n+1}$ is a solution of the equation $A_{n}f=0$. The Darboux
transformation of the operator $A_{n}$ with respect to the solution $f_{n}$ is
the operator
\[
\frac{d^{2}}{dx^{2}}-2\Big(\frac{f_{n}^{\prime}(x)}{f_{n}(x)}\Big)^{2}%
+\frac{n(n+1)}{(x+1)^{2}}=\frac{d^{2}}{dx^{2}}-\frac{(n+1)(n+2)}{(x+1)^{2}},
\]
i.e., exactly the operator $A_{n+1}$. If we know $\mathbf{K}_{n}(x,t;-n)$ for
the operator $A_{n}$, by \eqref{KmainChangeOfH} we compute the kernel
$\mathbf{K}_{n}(x,t;n+1)$ corresponding to the solution $f_{n}(x)$ and by
Theorem \ref{Th T2Volterra} we may calculate the kernel $\mathbf{K}%
_{n+1}(x,t;-n-1)$. Careful analysis shows that we have to integrate only
polynomials in all integrals involved, so the described procedure can be
performed up to any fixed $n$.
\end{ex}

\begin{ex}
Consider the Schr\"{o}dinger equation
\begin{equation}
u^{\prime\prime}+2\sech^{2}(x)\,u=u. \label{ExampleSoliton}%
\end{equation}
This equation appears in soliton theory and as an example of a reflectionless
potential in the one-dimensional quantum scattering theory (see, e.g.
\cite{Lamb}). Equation \eqref{ExampleSoliton} can be obtained as a result of
the Darboux transformation of the equation $u^{\prime\prime}=u$ with respect
to the solution $f(x)=\cosh x$. The transmutation operator for the operator
$A_{1}=\partial_{x}^{2}-1$ was calculated in \cite[Example 3]{CKT}. Its kernel
is given by the expression
\[
\mathbf{K}_{1}(x,t;0)=-\frac{1}{2}\frac{\sqrt{x^{2}-t^{2}}I_{1}(\sqrt
{x^{2}-t^{2}})}{x-t},
\]
where $I_{1}$ is the modified Bessel function of the first kind. Hence from
Theorem \ref{Th T2Volterra} we obtain the transmutation kernel for the
operator $A_{2}=\partial_{x}^{2}+2\sech^{2}x-1$
\[
\mathbf{K}_{2}(x,t;0)=\frac{1}{2\cosh(x)}\int_{-t}^{x}\bigg(  \frac
{I_{0}(\sqrt{s^{2}-t^{2}})t}{s-t}+\frac{\sqrt{s^{2}-t^{2}}I_{1}(\sqrt
{s^{2}-t^{2}})}{(s-t)^{2}}\bigg)  \cosh s\,ds.
\]

\end{ex}

\section{Transmutation operator for the one-dimensional Dirac equation with a
Lorentz scalar potential\label{Sect Transmutation Dirac}}

One-dimensional Dirac equations with Lorentz scalar potentials are widely
studied (see, for example, \cite{Casahorran, Chen, Hiller, Ho, JP, KP, KhR,
NT, RV, Rukeng} and \cite{NPS} for intertwining techniques for them).

According to \cite{NT} the Dirac equation in one space dimension with a
Lorentz scalar potential can be written as
\begin{align}
(\partial_{x}+m+S(x))\Psi_{1}  &  =E\Psi_{2},\label{Dirac1}\\
(-\partial_{x}+m+S(x))\Psi_{2}  &  =E\Psi_{1}, \label{Dirac2}%
\end{align}
where $m$ ($m>0$) is the mass and $S(x)$ is a Lorentz scalar. Denote
$\eta=m+S$ and write the system \eqref{Dirac1}, \eqref{Dirac2} in a matrix
form as
\[%
\begin{pmatrix}
\partial_{x}+\eta & 0\\
0 & \partial_{x}-\eta
\end{pmatrix}%
\begin{pmatrix}
\Psi_{1}\\
\Psi_{2}%
\end{pmatrix}
=E%
\begin{pmatrix}
0 & 1\\
-1 & 0
\end{pmatrix}%
\begin{pmatrix}
\Psi_{1}\\
\Psi_{2}%
\end{pmatrix}
. \label{DiracMatrix}%
\]
In order to apply the results on the transmutation operators and
factorizations \eqref{A1factor}, \eqref{A2factor} we consider a function $f$
such that
\[
\frac{f^{\prime}(x)}{f(x)}=-\eta=-m-S(x).
\]
We can take $f(x)=\exp\left(  -\int_{0}^{x}(m+S(s))\,ds\right)  $, then
$f(0)=1 $ and $f$ does not vanish. Suppose the operators $\mathbf{T}_{1}$ and
$\mathbf{T}_{2}$ are transmutations for the operators $A_{1}=\big(\partial
_{x}+\frac{f^{\prime}}{f}\big)\big(\partial_{x}-\frac{f^{\prime}}{f}\big)$ and
$A_{2}=\big(\partial_{x}-\frac{f^{\prime}}{f}\big)\big(\partial_{x}%
+\frac{f^{\prime}}{f}\big)$ respectively (corresponding to functions $f$ and
$1/f$ in the sense of Proposition \ref{ThMapsSolutions}). As was shown in
\cite{KrT2012} with the use of commutation relations \eqref{CommutT1dx} and
\eqref{CommutT2dx}, the operator
\[
\mathbf{T}=%
\begin{pmatrix}
\mathbf{T}_{1} & 0\\
0 & \mathbf{T}_{2}%
\end{pmatrix}
\]
transmutes any solution $%
\begin{pmatrix}
u_{1}\\
u_{2}%
\end{pmatrix}
$ of the system
\begin{align}
u_{1}^{\prime}  &  =Eu_{2}\label{DiracTriv1}\\
u_{2}^{\prime}  &  =-Eu_{1} \label{DiracTriv2}%
\end{align}
into the solution $%
\begin{pmatrix}
\Psi_{1}\\
\Psi_{2}%
\end{pmatrix}
$ of the system \eqref{Dirac1},\eqref{Dirac2} with the initial conditions
$\Psi_{1}(0)=u_{1}(0)$, $\Psi_{2}(0)=u_{2}(0)$. And vice versa if $\begin{pmatrix}
\Psi_{1}\\
\Psi_{2}%
\end{pmatrix}$ is a solution of the system \eqref{Dirac1},\eqref{Dirac2}, then the
operator $\Bigl(%
\begin{smallmatrix}
\mathbf{T}_{1}^{-1} & 0\\
0 & \mathbf{T}_{2}^{-1}%
\end{smallmatrix}
\Bigr)$ transmutes it into the solution $
\begin{pmatrix}
u_{1}\\
u_{2}%
\end{pmatrix}$ of \eqref{DiracTriv1},\eqref{DiracTriv2} such that $u_{1}(0)=\Psi
_{1}(0)$, $u_{2}(0)=\Psi_{2}(0)$.

Consider two systems of functions $\{\varphi_{k}\}_{k=0}^{\infty}$ and
$\{\psi_{k}\}_{k=0}^{\infty}$ constructed from the function $f$ by
\eqref{phik} and \eqref{psik}. The general solution of the system
\eqref{DiracTriv1},\eqref{DiracTriv2} is given by
\begin{align*}
u_{1}  &  = C_{1}v_{1}+C_{2}v_{2}\\
u_{2}  &  = C_{2}v_{1}-C_{1}v_{2},
\end{align*}
where $C_{1}$ and $C_{2}$ are arbitrary constants and
\begin{align*}
v_{1}(x)  &  =\cos Ex = \sum_{k=0}^{\infty}\frac{(-1)^kE^{2k}}{(2k)!}x^{2k},\\
v_{2}(x)  &  =\sin Ex = \sum_{k=0}^{\infty}\frac{(-1)^kE^{2k+1}}{(2k+1)!}x^{2k+1}.
\end{align*}
From \eqref{T1x^k} and \eqref{T2x^k} we see that the general solution of the
one-dimensional Dirac system \eqref{Dirac1},\eqref{Dirac2} has the form
\begin{align*}
\Psi_{1}  &  =C_{1}\sum_{k=0}^{\infty}\frac{(-1)^kE^{2k}}{(2k)!}\varphi_{2k}+
C_{2}\sum_{k=0}^{\infty}\frac{(-1)^kE^{2k+1}}{(2k+1)!}\varphi_{2k+1},\\
\Psi_{2}  &  =C_{2}\sum_{k=0}^{\infty}\frac{(-1)^kE^{2k}}{(2k)!}\psi_{2k}- C_{1}%
\sum_{k=0}^{\infty}\frac{(-1)^kE^{2k+1}}{(2k+1)!}\psi_{2k+1}.
\end{align*}

\begin{rem}
It is possible to consider the two- or three-dimensional Dirac system and to
construct the transmutation operator for it under some conditions on the
potential. But the techniques involved, such as bicomplex numbers,
pseudoanalytic function theory, Vekua equation and formal powers go well
beyond the scope of the present article. We refer interested readers to the
recent paper \cite{CKM}.
\end{rem}

\end{document}